\documentclass[a4paper,UKenglish,cleveref,pdfa]{lipics-v2021}
\nolinenumbers
\hideLIPIcs
\usepackage{booktabs}
\crefname{enumi}{}{}\Crefname{enumi}{}{}
\crefformat{enumi}{(#2#1#3)}\crefrangeformat{enumi}{(#3#1#4--#5#2#6)}
\usepackage{mathtools}
\usepackage{tikz}
\usetikzlibrary{calc,positioning}
\usepackage{tikz-cd}
\usepackage{xspace}

\newtheorem{fact}[theorem]{Fact}

\newcommand{\stacked}[1]{%
  \ifmmode\begin{array}[t]{@{}l@{}}\else\begin{tabular}[t]{@{}l@{}}\fi%
  #1
  \ifmmode\end{array}\else\end{tabular}\fi%
}
\newcommand*{\powerset}{\mathcal{P}}

\newcommand*{\lequiv}{\leftrightarrow}
\DeclareMathOperator{\dom}{dom}
\DeclareMathOperator{\rng}{rng}
\newcommand*{\bigland}{\bigwedge}
\newcommand*{\biglor}{\bigvee}

\newcommand*{\pos}[1]{\langle#1\rangle} 
\newcommand*{\nec}[1]{[#1]}
\newcommand*{\at}[1]{@_{#1}\,}
\newcommand*{\store}[1]{{\downarrow}#1\,{\cdot}\,}
\newcommand*{\Forall}[1]{\forall #1\,{\cdot}\,}
\newcommand*{\Exists}[1]{\exists #1\,{\cdot}\,}

\newcommand*{\PL}{\mathsf{PL}}
\newcommand*{\HPL}{\mathsf{HPL}}
\newcommand*{\HDPL}{\mathsf{HDPL}}
\newcommand*{\DPL}{\mathsf{DPL}}

\newcommand*{\Cat}{\mathbb{C}\mathsf{at}}
\newcommand*{\Set}{\mathbb{S}\mathsf{et}}

\newcommand*{\Sig}{\mathsf{Sig}}
\newcommand*{\Mod}{\mathsf{Mod}}
\newcommand*{\Sen}{\mathsf{Sen}}
\newcommand*{\Prop}{\mathtt{Prop}}

\newcommand*{\frag}{\mathcal{L}}
\newcommand*{\sop}{\mathcal{O}}
\newcommand{\Act}{{\mathcal{A}}}
\newcommand{\act}{{\mathfrak{a}}}
\newcommand{\str}[1]{\mathfrak{#1}}
\newcommand*{\red}[1]{\mathord{|}#1}
\newcommand*{\ext}[2]{^{#1\leftarrow#2}}

\newcommand*{\Abelard}{$\forall$belard\xspace}
\newcommand*{\Heloise}{$\exists$loise\xspace}
\newcommand*{\EF}{Ehrenfeucht-Fra\"iss\'{e}\xspace}
\newcommand*{\FH}{Fra\"iss\'{e}-Hintikka\xspace}
\newcommand*{\tr}{\mathit{tr}}
\newcommand*{\rt}{\mathit{root}}
\newcommand*{\lb}{\mathit{lb}}
\NewDocumentCommand{\exwins}{}{\approx}
\newcommand*{\bfsys}{\mathcal{I}}
\NewDocumentCommand{\bfequiv}{}{\equiv}
\NewDocumentCommand{\eequiv}{}{\equiv}

\usepackage{scalerel}

\newcommand{\bbsemicolon}{%
  \scalerel*{%
    \hbox{\usefont{U}{bbold}{m}{n} ;}%
  }{;}%
}
\newcommand{\comp}{\mathbin{\bbsemicolon}}
\tikzset{%
  world/.style={circle, ball color=white, },
  sig/.style={rounded corners, fill=gray!10, },
  model/.style={rounded corners, fill=gray!10, },
  move/.style={rounded corners, fill=gray!10, },
}

\usepackage{pdfcomment}
\makeatletter\let\endnote\@undefined\makeatother
\newif\ifshowednotes\showednotestrue
\newcommand*{\ednoteauthor}{EdNote}
\newcommand*{\ednotecomment}{No comment.}
\newcommand*{\myenotezwritemark}[1]{\leavevmode\marginpar{\pdftooltip{\footnotesize\ednoteauthor(#1)}{\ednotecomment}}}
\usepackage{enotez}
\setenotez{backref=true}
\setenotez{list-name={EdNotes}}
\setenotez{mark-cs={\myenotezwritemark}}
\newcommand{\ednote}[2][EdNote]{%
  \ifshowednotes%
    \renewcommand*{\ednoteauthor}{#1}%
    \renewcommand*{\ednotecomment}{#2}%
    \bgroup%
      \fontsize{6pt}{6pt}\selectfont%
      \endnote{\ifthenelse{\equal{#1}{EdNote}}{#2}{#1: #2}}%
    \egroup%
  \fi%
}
\ifshowednotes%
  \AtEndDocument{\printendnotes}%
\fi
\title{Hybrid-Dynamic \EF Games}
\author{Guillermo Badia}%
{The University of Queensland}%
{g.badia@uq.edu.au}%
{}%
{}
\author{Daniel G\u{a}in\u{a}}%
{Kyushu University}%
{daniel@imi.kyushu-u.ac.jp}%
{}%
{}
\author{Alexander Knapp}%
{Universität Augsburg}%
{alexander.knapp@uni-a.de}%
{https://orcid.org/0000-0002-4050-3249}%
{}
\author{Tomasz Kowalski}%
{Jagiellonian University \and La Trobe University \and The  University of Queensland}%
{tomasz.s.kowalski@uj.edu.pl}%
{https://orcid.org/0000-0003-3095-4043}%
{}
\author{Martin Wirsing}%
{Ludwig-Maximilians-Universität München}%
{wirsing@lmu.de}%
{}%
{}
\authorrunning{G.~Badia, D.~Găină, A.~Knapp, T.~Kowalski, and M.~Wirsing}
\Copyright{G.~Badia, D.~Găină, A.~Knapp, T.~Kowalski, and M.~Wirsing}

\begin{document}


\maketitle

\begin{abstract}
\EF games provide means to characterize elementary
equivalence for first-order logic, and by standard translation also
for modal logics. We propose a novel
generalization of \EF games to hybrid-dynamic logics which is direct and fully modular:
parameterized by the features of the hybrid language we wish to include,
for instance, the modal and hybrid language operators as well as first-order
existential quantification. We use these games to establish a new
modular Fra\"iss\'{e}-Hintikka theorem for hybrid-dynamic propositional logic and its
various fragments. We study the relationship between countable game equivalence
(determined by countable \EF games) and bisimulation (determined by countable
back-and-forth systems).  
In general, the former turns out to be weaker than the latter, but under certain
conditions on the language, the two coincide. As a corollary we obtain an analogue
of the Hennessy-Milner theorem.  We also prove that for
reachable image-finite Kripke structures elementary equivalence implies
isomorphism. 
\end{abstract}


\section{Introduction}

Hybrid logics and, in particular, hybrid-dynamic logics, are expressive modal logics well-suited for describing behavioral dynamics: \emph{if something holds at a certain state, then something else holds at some state accessible from it}. Hybrid dynamics logics have been applied to the specification and modelling
of reactive and event/data-based systems
(see, e.g., ~\cite{DBLP:journals/tcs/MadeiraBHM18,DBLP:conf/fase/RosenbergerKR22,HKM21}).
The expressivity of dynamic-propositional logic ($\DPL$) extends beyond its hybrid counterpart at the expense of important logical properties such as compactness~\cite{Verbrugge02}.
Lesser-known, perhaps, is the fact that hybrid-dynamic propositional logic ($\HDPL$) is sufficiently expressive to describe finite models (see \cref{ex:finite}). 

\EF (EF) games characterize elementary equivalence in first-order
logic~\cite{hodges93}.  
In this paper, we propose a modular notion of EF games for $\HDPL$ and its fragments, designed to accommodate various forms of quantification encountered in modal logics. 
The advantages of modularity in the context of hybrid logics are thoroughly demonstrated in~\cite{BGKKW25}.
We apply this framework to establish a \FH Theorem~\cite{hodges93} for hybrid-dynamic logics and to study the relationship between countable game equivalence (determined by countable EF games), \emph{$\omega$-bisimulations} and \emph{back-and-forth systems} \cite{ArecesBM01}.
In an EF game, the players' moves correspond to various types of (hybrid) quantification, including possibility over structured actions, first-order quantification, and \emph{store}, which names the current state.
The mathematical structure supporting this formalization is the \emph{gameboard
  tree}~\cite{gai-fraisse}, which we extend with edge labels.  
The role of such trees is comparable to that of a chessboard in chess.  
Gameboard trees can also be seen as representing the quantifier ranks of sentences, as nodes correspond to signatures, and edges represent labeled extensions of signatures with variables.

In this paper, we examine both versions of EF games: the finite and the countably infinite.
The winning strategies in finite EF games, played over gameboard trees, are precisely described by so-called \emph{game sentences}, which are defined over such trees.
In the literature, the result stating that a winning strategy in a finite EF game is exactly described by a game sentence is known as the \FH Theorem (\cref{th:fht}); it holds for all fragments of hybrid-dynamic propositional logic, including dynamic propositional logic.
To the best of our knowledge, no proof of the \FH Theorem for hybrid or dynamic logics appears in the existing literature.
A direct consequence of this theorem is a characterization of hybrid-dynamic elementary equivalence in terms of EF games (\cref{cor:fh}), asserting that finite EF games yield an equivalence that aligns with elementary equivalence: pointed models cannot be distinguished by sentences.
We further show that the equivalence induced by countably infinite EF games coincides with $\omega$-bisimilarity and, under certain conditions, with back-and-forth equivalence.
As applications, we prove an analogue of the Hennessy–Milner theorem and show that, for rooted image-finite models, elementary equivalence coincides with isomorphism.

Several variations of EF games have been proposed in the literature for hybrid logics, but not for hybrid-dynamic logics.
For example, EF games for hybrid temporal logic are introduced in~\cite{abramsky_et_al:LIPIcs.MFCS.2022.7}, while EF games for hybrid computational tree logic are proposed in~\cite{KERNBERGER2020362}.
Notably, the characterizations of hybrid elementary equivalence in these works are established directly by induction on the complexity of hybrid sentences, rather than as consequences of a \FH Theorem.

This work incorporates several features that shape its approach to EF games:
(a)~First, it is grounded in category theory, which provides a foundational framework for constructing EF games, including the notion of gameboard trees. 
In this setting, Spoiler/\Abelard's moves are suitably restricted, enabling the definition of game sentences --- formulas that precisely characterize the allowable moves in the game. These game sentences are then used to prove that elementary equivalence coincides with the existence of a winning strategy for Duplicator/\Heloise.
(b)~Second, the use of gameboard trees enables a modular approach to EF games. By selectively adding or removing conditions --- corresponding to the sentence operators under consideration --- one can adapt the framework to specific logical fragments. This flexibility is particularly valuable in applications where the full expressive power of hybrid-dynamic logic is unnecessary, and where attention may instead be limited to a fragment such as hybrid propositional logic ($\HPL$).

Other frameworks also support modular reasoning. 
For example, bisimulations --- commonly regarded as the modal counterpart to EF games \cite{Stirling1999} --- exhibit a modular structure. We show that $\omega$-bisimulations correspond to winning strategies for \Heloise in countably infinite EF games. 
In contrast, back-and-forth systems --- typically seen as formalizations of \Heloise's strategy \cite{ArecesBM01} --- tend to be stronger. Their equivalence with EF games holds only when the logical fragment is closed under certain sentence operators. This suggests that the back-and-forth approach may implicitly assume the availability of these operators in the language, which could affect its adaptability across different logical fragments.
Each approach has its own strengths and limitations, but a choice must ultimately be made based on the intended use and future development.
In this context, our approach to EF games --- grounded in gameboard trees and characterized by game sentences --- offers a pathway toward mechanization and the development of tools tailored to specific logical fragments, especially in domains where limited expressiveness is both sufficient and desirable.

\subparagraph{Structure of the paper.}%
We first introduce the logical framework of Hybrid-Dynamic Propositional Logic
($\HDPL$) in \cref{sec:hdpl} as a family of languages parameterized by sentence
constructors.  In \cref{sec:EFG}, we describe finite \EF games for $\HDPL$ and
its fragments and prove a general, modular \FH theorem for characterizing
elementary equivalence.  We move to infinite \EF games in
\cref{sec:countable-EFG} and relate these games to $\omega$-bisimulations and
back-and-forth systems; in particular, for image-finite models, we obtain a
Hennessy-Milner for the quantifier-free fragment of $\HPL$ and we show that rooted, elementarily equivalent models are isomorphic.  We conclude in \cref{sec:conclusions}.

\section{Hybrid-Dynamic Propositional Logic (\texorpdfstring{$\HDPL$}{HDPL})}\label{sec:hdpl}

\subsection{Signatures} 
The signatures are of the form $(\Sigma,\Prop)$,
where $\Sigma=(F,P)$ is a single-sorted first-order signature consisting of a
set of constants $F$ called nominals and a set of binary relation symbols $P$,
and $\Prop$ is a set of propositional symbols.  We let $\Delta$ range over
signatures of the form $(\Sigma, \Prop)$ as described above.  Similarly, for any
index $i$, we let $\Delta_i$ range over signatures of the form
$(\Sigma_i,\Prop_i)$, where $\Sigma_i$ is a single-sorted first-order signature
of the form $(F_i,P_i)$ defined similarly as above.  A signature morphism $\chi
: \Delta_1\to\Delta_2$ consists of a first-order signature morphism $\chi :
\Sigma_1\to \Sigma_2$ and a function $\chi : \Prop_1\to\Prop_2$.  The
\emph{extension} of $\Delta$ by a fresh nominal $k$ is denoted $\Delta[k]$
yielding the inclusion $\Delta \hookrightarrow \Delta[k]$.
We denote by $\Sig^\HDPL$ the category of $\HDPL$ signatures.

\subsection{Models} 
The models defined over a signature $\Delta$ are
standard Kripke structures $\str{M} = (W,M)$ such that: 
\begin{enumerate}
\item $W$ is a first-order structure over $\Sigma$. 
We denote by $|\str{M}|$ the universe of $W$ and we call the elements of $|\str{M}|$ states, possible worlds or nodes.
The accessibility relations consist of the interpretations of the binary relation symbols from $P$ into $W$, that is, $\lambda^W$ is an accessibility relation for each binary relation symbol $\lambda\in P$.
\item $M : |\str{M}| \to |\Mod^\PL(\Prop)|$ is a mapping from the set of states
$|\str{M}|$ to the class of propositional logic models $|\Mod^\PL(\Prop)|$,
i.\,e., subsets of propositional symbols.  
\end{enumerate}
We let $\str{M}$ and
$\str{N}$ range over Kripke structures of the form $(W,M)$ and $(V,N)$,
respectively.  Similarly, for any index $i$ we let $\str{M}_i$ and
$\str{N}_i$ range over Kripke structures of the form $(W_i,M_i)$ and
$(V_i,N_i)$, respectively.  A homomorphism $h:\str{M} \to \str{N}$
between two Kripke structures $\str{M}$ and $\str{N}$ is a first-order
homomorphism $h : W \to V$ such that $h(M(w))\subseteq N(h(w))$ for all states
$w \in |\str{M}|$.

\begin{fact}
For any signature $\Delta$, the $\Delta$-homomorphisms form a category
$\Mod^\HDPL(\Delta)$ under the obvious composition as many-sorted functions.
\end{fact}

For any signature morphism $\chi : \Delta_1 \to \Delta_2$ the reduct functor
${-}\red{\chi} : \Mod^\HDPL(\Delta_2) \to \Mod^\HDPL(\Delta_1)$ is defined as
follows: The reduct $\str{M}\red{\chi}$ of the $\Delta_2$-model $\str{M}$ is
$(W\red{\chi}, M\red{\chi})$, where $W\red{\chi}$ is the reduct of $W$ across
$\chi : \Sigma_1 \to \Sigma_2$ in first-order logic, and $(M\red{\chi})(w) =
M(w)\red{\chi} = \{ p \in \Prop_1 \mid \chi(p) \in M(w) \}$ is the reduct of
$M(w)$ across $\chi : \Prop_1 \to \Prop_2$ in propositional logic, for all
states $w \in |\str{M}|$.  The reduct $h\red{\chi}$ of a $\Delta_2$-homomorphism
$h : \str{M} \to \str{N}$ is the first-order homomorphism $h\red{\chi} :
W\red{\chi} \to V\red{\chi}$.  Since $h(M(w)) \subseteq N(w)$ for all states $w
\in |\str{M}|$, we get $(h\red{\chi})(M(w)\red{\chi}) \subseteq N(w)\red{\chi}$
for all states $w \in |\str{M}| = |\str{M}\red{\chi}|$, which means that
$h\red{\chi}$ is well-defined.  If $\chi : \Delta_1 \to \Delta_2$ is an
inclusion, we also write $\str{M}\red{\Delta_1}$ for
$\str{M}\red{\chi}$.

\begin{fact}
$\Mod^\HDPL : \Sig^\HDPL \to \Cat^{\mathrm{op}}$ is a functor with $\Mod^\HDPL(\chi)(h) =
h\red{\chi}$ for each signature morphism $\chi\colon\Delta_1\to\Delta_2$ and
each $\Delta_2$-homomorphism $h$, where $\Cat$ is the category of all
categories.
\end{fact}

\subsection{Actions}  
The set of actions $\Act^\HDPL(\Delta)$ over a signature
$\Delta$ is defined by the following grammar:
\begin{equation*}
\act\Coloneqq
\lambda \mid 
\act\cup\act\mid
\act \comp\act\mid
\act^*\ \text{,}
\end{equation*}
where $\lambda$ is a binary relation on nominals.
Actions are interpreted in Kripke structures $(W,M)$ as \emph{accessibility relations} between possible worlds.
This is done by extending the interpretation of binary relation symbols on nominals:
\begin{enumerate}[a)]
\item~$\lambda^\str{M}=\lambda^W$ for all binary relations $\lambda$ in $\Delta$,
\item~$(\act_1 \cup \act_2)^\str{M} = \act_1^\str{M}\cup \act_2^\str{M}$ (union),
\item~$(\act_1 \comp \act_2)^\str{M} = \act_1^\str{M} \comp \act_2^\str{M}$ (diagrammatic composition of relations), and
\item~$(\act^*)^\str{M}= (\act^\str{M})^*$ (reflexive \textit{\&} transitive closure).
\end{enumerate}

\subsection{Sentences}

Let $\{v_n\mid n\in\mathbb{N}\}$ be a set of variable names.
A variable for a signature $\Delta$ is a pair $x=\pos{v_n,\Delta}$, where $v_n$ is a variable name.
Notice that all variables $x=\pos{v_n,\Delta}$ are different from all symbols in $\Delta$.
The translation of $x$ along a signature morphism $\chi:\Delta\to\Delta'$ is $x'=\pos{v_n,\Delta'}$.
The set of sentences $\Sen^\HDPL(\Delta)$ over a signature $\Delta$ is defined
by the following grammar:
\begin{equation*}\textstyle
\phi \Coloneqq p \mid
k \mid
\bigland \Phi \mid
\neg \phi\mid
\pos{\act}\phi\mid 
\at{k}\phi\mid
\store{x}\phi_x \mid 
\Exists{x}\phi_x\ \text{,}
\end{equation*}
where $p$ is a propositional symbol, 
$k$ is a nominal, 
$\Phi$ is a finite set of sentences over $\Delta$,
$x$ is a variable for $\Delta$, 
$\act$ is an action over $\Delta$, 
and $\phi_x \in \Sen^{\HDPL}(\Delta[x])$.  
The sentence $\pos{\act}\phi$ is read as ``$\phi$ holds after $\act$'' (\emph{possibility}), 
$\at{k}\phi$ as ``$\phi$ holds at state $k$'' (\emph{retrieve}), and 
$\store{x}\phi_x$ as ``$\phi_x$ holds with the current state bound to $x$'' (\emph{store}).  
We will use the usual abbreviations $\biglor\Phi$ for $\neg(\bigland_{\phi\in\Phi} \neg\phi)$, $\nec{\act}\phi$ for
$\neg\pos{\act}\neg\phi$, and $\Forall{x}\phi_x$ for $\neg\Exists{x}\neg\phi_x$.
Each signature morphism $\chi : \Delta_1 \to \Delta_2$ induces a sentence
translation $\chi : \Sen(\Delta_1)\to \Sen(\Delta_2)$ that replaces, in an
inductive manner, in any $\Delta_1$-sentence $\phi$ each symbol from $\Delta_1$
with a symbol from $\Delta_2$ according to $\chi : \Delta_1\to \Delta_2$.

It is worth noting that the translation of a quantified sentence $\Exists{x}\phi_x\in\Sen(\Delta_1)$, where $x=\pos{v_n,\Delta_1}$, along $\chi:\Delta_1\to \Delta_2$ is $\Exists{x'}\chi'(\phi_x)$, where $x'=\pos{v_n,\Delta_2}$ and $\chi':\Delta_1[x]\to\Delta_2[x']$ is the extension of $\chi:\Delta_1\to\Delta_2$ which maps $x$ to its translation $x'$. 
\begin{center}
\begin{tikzcd}
\phi_x \ar[r,dotted,no head]& 
\Delta_1[x] \ar[r,"\chi'"]& 
\Delta_2[x'] & 
\chi'(\phi_x) \ar[l,dotted,no head] \\
\Exists{x}\phi_x\ar[r,dotted,no head]  & 
\Delta_1 \ar[r,"\chi",swap] \ar[u,hook]& 
\Delta_2 \ar[u,hook] & 
\Exists{x'}\chi'(\phi_x) \ar[l,dotted,no head]
\end{tikzcd}
\end{center}
\begin{fact}
$\Sen^\HDPL : \Sig^\HDPL\to\Set$ is a functor, where $\Set$ is the category of all sets.
\end{fact}

This approach proves useful in preventing clashes between variables and constants within a given signature.
Furthermore, substitutions can be defined cleanly, avoiding the need for side conditions.
When there is no danger of confusion (e.g., a variable name is among the elements of a signature), we identify a variable only by its name.
This means that signature inclusions $\Delta_1\subseteq \Delta_2$ determine inclusions of sets of sentences $\Sen(\Delta_1)\subseteq\Sen(\Delta_2)$.

\subsection{Local satisfaction relation}
The local satisfaction relation of a sentence $\phi$ over a model $\str{M} =
(W, M)$ over a signature $\Delta$ and a world $w \in |\str{M}|$ is defined by
induction on the structure of the sentence:
\begin{itemize}
\item $(\str{M},w)\models p$ if $M(w) \models p$ in propositional logic, that
is, $p \in M(w)$;

\item $(\str{M},w)\models k$ if $w = k^{\str{M}}$

\item $(\str{M},w)\models \bigland \Phi$ if $(\str{M},w)\models \phi$ for all $\phi\in\Phi$;

\item $(\str{M},w)\models \neg \phi$ if $(\str{M},w)\not\models \phi$;

\item $(\str{M},w)\models\at{k}\phi$ if $(\str{M},k^{\str{M}})\models\phi$;

\item $(\str{M},w) \models \pos{\act}\phi$ if 
$(\str{M}, v) \models \phi$ for some $v \in \act^{\str{M}}(w) \coloneqq \{ w' \in |\str{M}| \mid (w, w')
\in \act^\str{M} \}$;

\item $(\str{M},w)\models \store{x}\phi_x$ if
$(\str{M}\ext{x}{w},w)\models\phi_x$, 

where $\str{M}\ext{x}{w}$ is the unique expansion of $\str{M}$ to $\Delta[x]$ interpreting $x$ as $w$;

\item $(\str{M},w)\models \Exists{x}\phi_x$ if $(\str{M}\ext{x}{v}, w)\models\phi_x$ for some $v\in|\str{M}|$.
\end{itemize}
We call the pair $(\str{M},w)$ a \emph{pointed model}, where $w$ is the
\emph{current} (or \emph{initial}) state.

\begin{theorem}[Local satisfaction condition]\label{thm:sat-cond}
For all signature morphisms $\chi:\Delta_1\to\Delta_2$, 
all $\Delta_2$-models $\str{M}$,
all states $w\in|\str{M}|$,
all $\Delta_1$-sentences $\phi$,
we have
\begin{equation*}
(\str{M},w)\models\chi(\phi)$ iff $(\str{M}\red{\chi}, w)\models\phi
\ \text{.}
\end{equation*}
\end{theorem}

\Cref{thm:sat-cond} shows that $\HDPL$ is a stratified institution according to the definitions given in \cite{dia-ult-kripke,gai-godel}.  
For a proof of the local satisfaction condition one may refer to \cite{dia-qvh}.


As an example of the expressivity of $\HDPL$ we outline a method for characterizing finite linear orderings:

\begin{example}\label{ex:finite}
Let $\Delta$ be a signature with two nominals $k_1,k_2$ and one binary relation $\lambda$.
Here, $k_1$ and $k_2$ represent the minimal and maximal elements, respectively, while $\lambda$ represents the immediate successor relation in a discrete linear ordering:
\begin{itemize}
\item Let $\phi_1\coloneqq (\Exists{x!} \at{k_1} \pos{\lambda} x) \wedge (\Forall{y}\neg\at{y}\pos{\lambda}k_1)$, where $\Exists{x!}\gamma\coloneqq\Exists{x}\gamma\wedge \Forall{y}\gamma[y/x] \to \at{x}y$, indicating that there exists a unique $x$ such that $\gamma\in\Sen(\Delta[x])$ is satisfied.
Notice that $\phi_1$ asserts that $k_1$ has a unique successor and no predecessor, which --- in the context where all elements are connected by finite paths (see the definition of $\phi_4$) --- means that $k_1$ is the minimal element.

\item Let $\phi_2\coloneqq (\Exists{x!} \at{x} \pos{\lambda} k_2) \wedge (\Forall{y}\neg\at{k_2}\pos{\lambda}y)$, which states that $k_2$ has a unique predecessor and a no successor, which in this context means that $k_2$ is the maximal element.
\item Let $\phi_3\coloneqq \Forall{x} \neg \at{x}k_1\wedge \neg \at{x}k_2 \to (\Exists{y!}\at{y}\pos{\lambda}x) \wedge (\Exists{z!}\at{x}\pos{\lambda}z)$, 
which asserts that each element $x$, not being the minimum nor maximum, has a unique predecessor and successor.
\item Let $\phi_4\coloneqq \Forall{x,y} \at{x}\pos{\lambda^*}y\vee \at{y}\pos{\lambda^*}x$, which expresses that any elements $x$ and $y$ are connected by a finite path of edges labeled $\lambda$.
\end{itemize}
Therefore, the sentence $\phi\coloneqq\phi_1\wedge\dots\wedge\phi_4$ has only finite models. 
\end{example}

\subsection{Logical framework and related concepts}

The logical framework in which the results will be developed in this paper is an arbitrary fragment $\frag=(\Sig^\HDPL,\Sen^\frag,\Mod^\HDPL,\models)$ of $\HDPL$ which is closed under Boolean connectives.
This means that $\frag$ is obtained from $\HDPL$ by discarding
\begin{itemize}
\item some constructors for actions from the grammar which defines actions in $\HDPL$, and/or
\item some constructors for sentences from the grammar which defines sentences in $\HDPL$.
\end{itemize}
For example, 
one can discard all action constructors and develop results over hybrid propositional logic ($\HPL$); or
one can work with the quantifier-free version of $\HDPL$;
or one can discard retrieve, store, and existential quantification, and work with dynamic multi-modal propositional logic, $\DPL$.
For the sake of simplicity, we will drop the superscripts $\frag$ and $\HDPL$ from $\frag=(\Sig^\HDPL,\Sen^\frag,\Mod^\HDPL,\models)$ when there is no danger of confusion.

We make the following notational conventions:
\begin{itemize}
\item Let $\sop\subseteq \{\Diamond,@,\downarrow,\exists\}$ be the subset of the sentence constructors which belong to $\frag$.\footnote{$\Diamond\in\sop$ means that $\frag$ is closed under possibility over actions, but no assumption is made concerning the existence of action constructors. One or more constructors for actions can be discarded from $\frag$.}
\item Let $\Mod_p(\Delta)=\{(\str{M},w)\mid \str{M}\in|\Mod(\Delta)| \text{ and } w\in|\str{M}|\}$ be the class of pointed $\Delta$-models.
\item Let $\Mod_p(\Delta,\phi)=\{(\str{M},w)\in\Mod_p(\Delta)\mid (\str{M},w)\models\phi\}$ be the class of pointed $\Delta$-models which satisfy a sentence $\phi$.
\item Let $\Sen_b:\Sig\to\Set$ be the subfunctor of $\Sen:\Sig\to\Set$ which maps each signature $\Delta=(\Sigma,\Prop)$ to the set of basic sentences $F\cup \Prop$.
\end{itemize}

\begin{definition} [Elementary equivalence]
Let $\Delta$ be a signature.
\begin{itemize}
\item Two pointed $\Delta$-models $(\str{M},w)$ and $(\str{N},v)$ are $\frag$-\emph{elementarily equivalent}, in symbols, 

$(\str{M},w)\equiv^\frag(\str{N},v)$, when 
$(\str{M},w)\models\phi$ iff $(\str{N},v)\models\phi$, 
for all $\Delta$-sentences $\phi\in\Sen^{\frag}(\Delta)$.
\item Two $\Delta$-sentences $\phi_1$ and $\phi_2$ are \emph{semantically equivalent}  if 
they are satisfied by the same pointed models, that is, 
$\Mod_p(\Delta,\phi_1)=\Mod_p(\Delta,\phi_2)$.
\end{itemize}
\end{definition}

When there is no danger of confusion, we drop $\frag$ from the above notations.

\section{Finite \EF Games}\label{sec:EFG}

We propose a notion of EF game for hybrid-dynamic propositional logic and its fragments
by generalizing in a non-trivial way the notion of EF game from first-order
logic~\cite{hodges93}.  In this section, we are interested in characterizing
elementary equivalence of pointed models in terms of EF games.

\subsection{Gameboard trees} 
The EF games proposed in this paper are played on a gameboard tree between \Abelard and \Heloise exactly like the EF games defined in \cite{gai-fraisse}.  
The nodes are labeled by signatures, while the edges are labeled by sentence operators and are uniquely identified by their source and label.
Labels are introduced to account for moves in games corresponding to various types of quantification used in defining the language of $\HDPL$ (e.g., possibility of structured actions, store, or first-order quantification).
The edges of the gameboard trees defined in \cite{gai-fraisse} are unlabeled, since that work considers only one type of quantification --- namely, first-order quantification.
There are four types of edges classified by their label, which are depicted in \cref{fig:tree}. 
Each type of edge is discarded from the gameboard tree if it is not included in the language fragment $\frag$ of discourse.
\begin{figure}[!ht]
\begin{tikzpicture}[baseline=0pt, sibling distance = 2cm, scale=.8, level distance=2cm, edge from parent/.style={draw,-latex}]
\tikzstyle{every node}=[fill=gray!10,rounded corners]
\node {$\Delta$}[grow=up]
  child{node {$\Delta$}
        edge from parent node{$1$}
       }
  child{node {$\Delta[x]$}
        edge from parent node{$\exists$}
       }
  child{node {$\Delta[x]$}
        edge from parent node{$\downarrow$}
       }
  child{node{$\Delta$}
        edge from parent node{$\at{k}$}
       }
  child{node{$\Delta$}
        edge from parent node{$\pos{\act}$}  
       };
\end{tikzpicture}
\caption{Gameboard tree}
\label{fig:tree}
\end{figure}
\begin{description}
\item[$\Diamond\in\sop$:] 
The edge $\Delta \xrightarrow{\pos{\act}} \Delta$ connects two distinct nodes that are both labeled by the same signature $\Delta$, where $\act$ is an action. 
It may be regarded as a pair consisting of the identity signature morphism $1_\Delta:\Delta\to \Delta$ and the modal operator $\pos{\act}$.

\item[$@\in\sop$:]
The edge $\Delta \xrightarrow{\at{k}} \Delta$ connects distinct nodes that are labeled by the same signature $\Delta$, where $k$ is a nominal.  
It may be regarded as a pair consisting of the identity signature morphism  $1_\Delta:\Delta\to \Delta$ and the operator retrieve $\at{k}$.

\item[$\downarrow\ \in\sop$:] 
The edge $\Delta \xrightarrow{\downarrow} \Delta[x]$ connects distinct nodes labeled by the signatures $\Delta$ and $\Delta[x]$, where $x=\pos{v_0,\Delta}$ is a variable for $\Delta$.
It may be regarded as a pair consisting of a signature inclusion $\Delta\longrightarrow\Delta[x]$ and the operator store ${\downarrow}x$.

\item[$\exists\in\sop$:] 
The edge $\Delta \xrightarrow{\exists } \Delta[x]$ connects distinct nodes labeled by the signatures $\Delta$ and $\Delta[x]$,  
where $x=\pos{v_0,\Delta}$ is a variable for $\Delta$. 
It may be regarded as a pair  consisting of a signature inclusion $\Delta\longrightarrow\Delta[x]$ and the first-order quantifier $\exists{x}$.

\item [idle:]
The edge $\Delta\xrightarrow{1}\Delta$ connects distinct nodes that are labeled by the same signature $\Delta$.
It may be regarded as the identity signature morphism $1_\Delta:\Delta\to\Delta$.
The idle edge serves to construct complex gameboard trees from simpler components and to define game sentences that represent conjunctions.
\end{description}
If $\frag$ is closed under the possibility over actions, which in turn are not restricted, the number of edges from a given finite signature is infinite.
In our approach we restrict  \Abelard's choices to a finite set by playing the game on finite gameboard trees.
If one discards the action operators, then the game can be played on complete gameboard trees which are finite provided that the root signature is finite.\footnote{By a complete gameboard tree, we mean a gameboard tree that explores all possible moves, each labeled by a sentence operator, such that all subtrees rooted at its children have equal height.}
In this case, \Abelard's choices are not restricted in any way and the game is similar to the classical one.

\subsection{\EF games} 
The game starts with two pointed models $(\str{M},w)$ and $(\str{N},v)$ defined over the same signature $\Delta$, and a gameboard tree $\tr$ with $\rt(\tr) = \Delta$.
\Heloise loses if the following \emph{game property} is not satisfied:
\begin{equation*}
(\str{M},w)\models\phi \text{ iff } (\str{N},v)\models\phi \quad\text{for all basic sentences $\phi\in\Sen_b(\Delta)$.}
\end{equation*}
Otherwise, the game can continue and \Abelard can move one of the pointed models
along an edge of the gameboard tree.  Without loss of generality, we assume that
\Abelard picks up the first pointed model $(\str{M}, w)$.
\begin{description}
\item[$\Diamond\in\sop$:]  
A move along $\Delta \xrightarrow{\pos{\act}} \Delta$
means that the next state $w_1$ chosen by \Abelard is accessible from $w$ via
$\act^{\str{M}}$, his new pointed model becoming $(\str{M},w_1)$ with $w
\mathrel{\act^{\str{M}}} w_1$.  \Heloise needs to find a state $v_1$
accessible from $v$ via $\act^{\str{N}}$ such that for her resulting pointed
model $(\str{N},v_1)$ the game property holds again.
\begin{equation*}
\begin{tikzpicture}[baseline=0pt,scale=.8, transform shape]
\tikzstyle{every node}=[fill=gray!10,rounded corners]
\node (M) at (-0.9,0.4) {$(\str{M},w)$};
\node  (O) at (0,0) {$\Delta$}; 
\node (N) at (-0.85,-0.4) {$(\str{N},v)$};
\node (M')at (5,0.4) {$(\str{M},w_1)$};
\node (O')at (4,0) {$\Delta$}; 
\node (N') at (4.9,-0.4) {$(\str{N},v_1)$};
\draw [->] (O) to node[above]{$\pos{\act}$} (O');
\end{tikzpicture}
\end{equation*}
\item[$@\in\sop$:]
A move along an edge $\Delta \xrightarrow{\at{k}} \Delta$,
where $k$ is a nominal in $\Delta$,
means that \Abelard changes the current state $w$ to $k^{\str{M}}$.  
The only possible choice for \Heloise is $(\str{N}, k^{\str{N}})$.
The game continues with $(\str{M},k^{\str{M}})$ and $(\str{N}, k^{\str{N}})$ if the game property holds again.
\begin{equation*}
\begin{tikzpicture}[baseline=0pt,scale=.8, transform shape]
\tikzstyle{every node}=[fill=gray!10,rounded corners]
\node (M) at (-0.9,0.4) {$(\str{M},w)$};
\node (O) at (0,0) {$\Delta$}; 
\node (N) at (-0.85,-0.4) {$(\str{N},v)$};
\node (M')at (5,0.4) {$(\str{M},k^{\str{M}})$};
\node (O')at (4,0) {$\Delta$}; 
\node (N')at (4.95,-0.4) {$(\str{N},k^{\str{N}})$};
\draw [->] (O) to node[above]{$\at{k}$} (O');
\end{tikzpicture}
\end{equation*}
\item[$\downarrow\ \in\sop$:] 
A move along $\Delta \xrightarrow{\downarrow} \Delta[x]$ means that \Abelard names the current state $x$, turning $(\str{M}, w)$ into
$(\str{M}\ext{x}{w}, w)$.  
\Heloise can only name her current state $x$, changing $(\str{N}, v)$ into $(\str{N}\ext{x}{v}, v)$.  
The game continues if the game property holds for these new pointed models.
\begin{equation*}
\begin{tikzpicture}[baseline=0pt,scale=.8, transform shape]
\tikzstyle{every node}=[fill=gray!10,rounded corners]
\node (M) at (-0.9,0.4) {$(\str{M},w)$};
\node (O) at (0,0) {$\Delta$}; 
\node (N) at (-0.85,-0.4) {$(\str{N},v)$};
\node (M')at (5.4,0.4) {$(\str{M}\ext{x}{w}, w)$};
\node (O')at (4,0) {$\Delta[x]$}; 
\node (N') at (5.3,-0.4) {$(\str{N}\ext{x}{v}, v)$};
\draw [->] (O) to node[above]{$\downarrow$} (O');
\end{tikzpicture}
\end{equation*}

\item[$\exists\in\sop$:] 
A move along $\Delta \xrightarrow{\exists} \Delta[x]$ means that \Abelard names $x$ a new arbitrary state $w_1\in|\str{M}|$ without changing his current state.  
His pointed model becomes $(\str{M}\ext{x}{w_1}, w)$.  
\Heloise needs to match \Abelard's choice by naming $x$ a state $v_1\in|\str{N}|$.  
The game continues with the new pointed models $(\str{M}\ext{x}{w_1}, w)$ and $(\str{N}\ext{x}{v_1},v)$ if the game property holds again.
\begin{equation*}
\begin{tikzpicture}[baseline=0pt,scale=.8, transform shape]
\tikzstyle{every node}=[fill=gray!10,rounded corners]
\node (M) at (-0.9,0.4) {$(\str{M},w)$};
\node [fill=gray!10,rounded corners] (O) at (0,0) {$\Delta$}; 
\node (N) at (-0.85,-0.4) {$(\str{N},v)$};
\node (M')at (5.5,0.4) {$(\str{M}\ext{x}{w_1}, w)$};
\node (O')at (4,0) {$\Delta[x]$}; 
\node (N') at (5.4,-0.4) {$(\str{N}\ext{x}{v_1}, v)$};
\draw [->] (O) to node[above]{$\exists$} (O');
\end{tikzpicture}
\end{equation*}
  \item[idle:] A move along $\Delta\xrightarrow{1}\Delta$ does not change the
pointed models.  (These moves are needed for modularity and their role will be
clear after reading the proof of \cref{th:fht}(\ref{it:th:fht:3})).
\begin{equation*}
\begin{tikzpicture}[baseline=0pt,scale=.8, transform shape]
\tikzstyle{every node}=[fill=gray!10,rounded corners]
\node (M) at (-0.9,0.4) {$(\str{M},w)$};
\node [fill=gray!10,rounded corners] (O) at (0,0) {$\Delta$}; 
\node (N) at (-0.85,-0.4) {$(\str{N},v)$};
\node (M')at (5,0.4) {$(\str{M}, w)$};
\node (O')at (4,0) {$\Delta$}; 
\node (N') at (5,-0.4) {$(\str{N}, v)$};
\draw [->] (O) to node[above]{$1$} (O');
\end{tikzpicture}
\end{equation*}
\end{description}

\Heloise loses the game if the game property is not satisfied by the current
pair of pointed models; \Heloise wins the game if she can match any move made by
\Abelard such that the game property is satisfied.  We write $(\str{M},w)
\exwins_{\tr} (\str{N},v)$ if \Heloise has a winning strategy for all hybrid-dynamic EF
game played on the gameboard tree $\tr$.  Notice that if \Heloise has a winning strategy
over a gameboard tree $\tr$, then she has a winning strategy over any gameboard
tree $\tr'$ included in $\tr$.
Therefore, in $\HPL$, one can work only with complete gameboard trees, since the number of binary relations of a finite signature is obviously finite. 
On the other hand, in $\HDPL$, the set of actions is countably infinite for each finite signature with at least one binary relation. 
This implies that a complete gameboard tree would be infinitely branched.

\begin{example}\label{ex:loop}
Let $\Delta$ be a signature with no nominals, one binary relation $\lambda$, and one propositional symbol $p$.
Let $\str{M}$ and $\str{N}$ be the $\Delta$-models shown to the left and right, respectively, in the following diagram.\footnote{This example is from the term-rewriting literature: Both models are abstract rewriting systems that are locally confluent but not confluent. In this case, $p$ stands for the normal form property~\cite{DBLP:books/daglib/0092409}. These concepts play no further role in our study.}
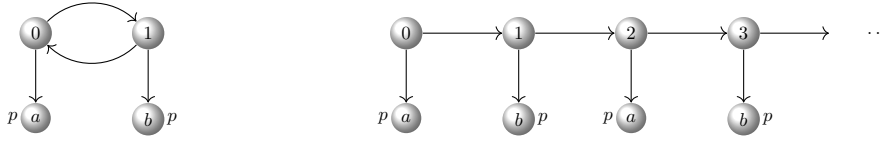
\begin{equation*}
\begin{tikzpicture}[baseline=0pt, sibling distance = 1cm, level distance=1.2cm, edge from parent/.style={draw,-latex},scale=.7, transform shape]
\begin{scope}
\node(C0)[circle,ball color=white,fill=gray!10]{$0$};
\node(C1)[circle,ball color=white,fill=gray!10,right=1.5cm of C0]{$1$};
\node(A)[circle,ball color=white,fill=gray!10,below=1cm of C0]{$a$};
\node(B)[circle,ball color=white,fill=gray!10,below=1cm of C1]{$b$};
\draw [->](C0) to[out=45, in = 135] (C1);
\draw [->](C1) to[out=-135, in = -45] (C0);
\draw [->](C0) to (A);
\draw [->](C1) to (B);
\node[left=-2pt of A]{$p$};
\node[right=-2pt of B]{$p$};
\end{scope}
\begin{scope}[xshift=7cm]
\node(C0)[circle,ball color=white,fill=gray!10]{$0$};
\node(C1)[circle,ball color=white,fill=gray!10,right=1.5cm of C0]{$1$};
\node(A0)[circle,ball color=white,fill=gray!10,below=1cm of C0]{$a$};
\node(B1)[circle,ball color=white,fill=gray!10,below=1cm of C1]{$b$};
\node(C2)[circle,ball color=white,fill=gray!10,right=1.5cm of C1]{$2$};
\node(C3)[circle,ball color=white,fill=gray!10,right=1.5cm of C2]{$3$};
\node(A2)[circle,ball color=white,fill=gray!10,below=1cm of C2]{$a$};
\node(B3)[circle,ball color=white,fill=gray!10,below=1cm of C3]{$b$};
\draw [->](C0) to (C1);
\draw [->](C0) to (A0);
\draw [->](C1) to (B1);
\draw [->](C1) to (C2);
\draw [->](C2) to (C3);
\draw [->](C2) to (A2);
\draw [->](C3) to (B3);
\node[left=-2pt of A0]{$p$};
\node[right=-2pt of B1]{$p$};
\node[left=-2pt of A2]{$p$};
\node[right=-2pt of B3]{$p$};
\node(C4)[right=1.3cm of C3]{};
\draw [->](C3) to (C4);
\node[right=10pt of C4]{$\dots$};
\end{scope}
\end{tikzpicture}
\end{equation*}
\begin{enumerate}
\item If $\frag$ is $\DPL$, then $(\str{M},0)$ and $(\str{N},0)$ are $\frag$-elementarily equivalent.
Since there are no nominals and no way to name them in the absence of store, \Heloise has a winning strategy for all EF games played over any gameboard tree.

\item If $\frag$ is $\HPL$, then $(\str{M},0) \not \approx_\tr (\str{N},0)$ for any complete gameboard tree $\tr$ of height greater or equal than $3$.
The sequence of moves that ensures \Abelard's winning is depicted in the following diagram.
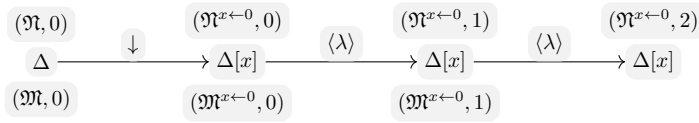
\begin{equation*}
\begin{tikzpicture}[baseline=0pt, scale=.8, transform shape]
\node(A0)[sig]{$\Delta$};
\node[model, above=2pt of A0]{$(\str{N},0)$};
\node[model, below=2pt of A0]{$(\str{M},0)$};
\node(A1)[sig, right=2.5cm of A0]{$\Delta[x]$};
\node[model, above=2pt of A1]{$(\str{N}\ext{x}{0},0)$};
\draw[->] (A0) to node[sig,above]{$\downarrow$} (A1) ;
\node[model, below=2pt of A1]{$(\str{M}\ext{x}{0},0)$};
\node(A2)[sig, right=2.5cm of A1]{$\Delta[x]$};
\node[model, above=2pt of A2]{$(\str{N}\ext{x}{0},1)$};
\draw[->] (A1) to node[sig,above]{$\pos{\lambda}$} (A2);
\node[model, below=2pt of A2]{$(\str{M}\ext{x}{0},1)$};
\node(A3)[sig, right=2.5cm of A2]{$\Delta[x]$};
\node[model, above=2pt of A3]{$(\str{N}\ext{x}{0},2)$};
\draw[->] (A2) to node[sig,above]{$\pos{\lambda}$} (A3);
\end{tikzpicture}
\end{equation*}
In the third round, \Heloise's only options are to move back to $0$ or to move
to $b$.  In both cases the game property is not satisfied.
\end{enumerate}
\end{example}

\begin{example}[{\cite[Ex.~5.13]{aceto-et-al:2007}}]\label{ex:infinite-branch}
Let $\Delta$ be the signature comprising of one binary relation $\lambda$.
Let $\str{M}$ be the model depicted to the left of the following diagram, which is a countably infinitely branched tree with the root $0$.
Let $\str{N}$ be the model depicted to the right of the following diagram, which is obtained from $\str{M}$ by adding a new branch of countably infinite length.
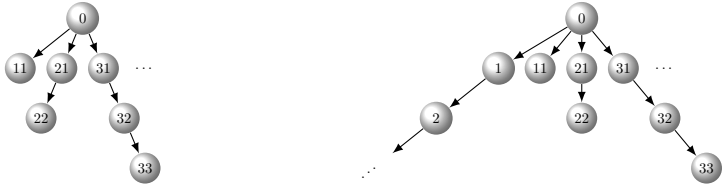
\begin{equation*}
\begin{tikzpicture}[baseline=0pt, sibling distance = 1cm, level distance=1.2cm, edge from parent/.style={draw,-latex},scale=.55, transform shape]
\begin{scope}
\node[circle,ball color=white,]{$~0~$}[grow=down]
  child{node[circle,ball color=white] {$11$}
        edge from parent
       }
  child{node[circle,ball color=white]{$21$}
        edge from parent
        child{ node[circle,ball color=white]{$22$} }
        child{ edge from parent [draw=none] }
       }
  child{node[circle,ball color=white]{$31$}
        edge from parent
        child{ edge from parent [draw=none] }
        child{ node[ball color=white,circle] {$32$}
               edge from parent
               child{ edge from parent [draw=none] }
               child{node[circle,ball color=white]{$33$}}
             } 
       }
  child{node{$\dots$}
        edge from parent[draw=none]
       };
\end{scope}
\begin{scope}[xshift=12cm]
\node[circle,ball color=white, fill=gray!10]{$~0~$}[grow=down]
  child{node[circle,ball color=white, fill=gray!10]{$~1~$}
        edge from parent
        child{node[circle,ball color=white, fill=gray!10]{$~2~$}
              edge from parent
              child{ node{\reflectbox{~~$\ddots$}}}
              child{ edge from parent [draw=none] }
              child{ edge from parent [draw=none] }
              child{ edge from parent [draw=none] }
             }
        child{edge from parent [draw=none]}
        child{edge from parent [draw=none]}
        child{edge from parent [draw=none]}     
       }
  child{node[circle,ball color=white, fill=gray!10] {$11$}
        edge from parent
       }
  child{node[circle,ball color=white, fill=gray!10]{$21$}
        edge from parent
        child{ node[circle,ball color=white, fill=gray!10]{$22$} }
       }
  child{node[circle,ball color=white, fill=gray!10]{$31$}
        edge from parent
        child{ edge from parent [draw=none] }
        child{ edge from parent [draw=none] }
        child{ node[ball color=white, fill=gray!10,circle] {$32$}
               edge from parent
               child{ edge from parent [draw=none] }
               child{ edge from parent [draw=none] }
               child{node[circle,ball color=white, fill=gray!10]{$33$}}
             } 
       }
  child{node{$\dots$}
        edge from parent[draw=none]
       };
\end{scope}
\end{tikzpicture}
\end{equation*}
\begin{enumerate}
\item  In all fragments of $\HPL$, \Heloise has a winning strategy for all EF games played over a finite gameboard tree starting with $(\str{M},0)$ and $(\str{N},0)$.
\item In any fragment of $\HDPL$ closed under Boolean connectives and possibility over structured actions, \Heloise loses the game played over $\Delta\xrightarrow{\pos{\lambda}}\Delta\xrightarrow{\pos{\lambda^*}}\Delta\xrightarrow{\pos{\lambda}}\Delta$ starting with $(\str{M},0)$ and $(\str{N},0)$:
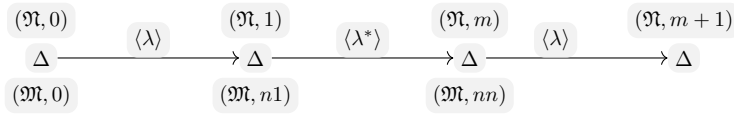
\begin{equation*}
\begin{tikzpicture}[baseline=0pt, scale=.8, transform shape]
\node[model, above=2pt of A0]{$(\str{N},0)$};
\node(A0)[sig]{$\Delta$};
\node[model, below=2pt of A0]{$(\str{M},0)$};
\node(A1)[sig, right=3cm of A0]{$\Delta$};
\node[model, above=2pt of A1]{$(\str{N},1)$};
\draw[->] (A0) to node[sig,above]{$\pos{\lambda}$} (A1) ;
\node[model, below=2pt of A1]{$(\str{M},n1)$};
\node(A2)[sig, right=3cm of A1]{$\Delta$};
\node[model, above=2pt of A2]{$(\str{N},m)$};
\draw[->] (A1) to node[sig,above]{$\pos{\lambda^*}$} (A2);
\node[model, below=2pt of A2]{$(\str{M},nn)$};
\node(A3)[sig, right=3cm of A2]{$\Delta$};
\node[model, above=2pt of A3]{$(\str{N},m+1)$};
\draw[->] (A2) to node[sig,above,xshift=-10pt]{$\pos{\lambda}$} (A3);
\end{tikzpicture}
\end{equation*}
First, \Abelard moves $(\str{N},0)$ along $\Delta\xrightarrow{\pos{\lambda}}\Delta$ to obtain $(\str{N},1)$.
A move from \Heloise results in $(\str{M},n1)$, where $n$ is a natural number greater than $0$.
Then \Abelard moves $(\str{M},n1)$ along  $\Delta\xrightarrow{\pos{\lambda^*}}\Delta$ to obtain $(\str{M},nn)$.
\Heloise can move $(\str{N},1)$   along  $\Delta\xrightarrow{\pos{\lambda^*}}\Delta$ to get $(\str{N},m)$, where $m$ is any natural number greater than $1$. 
For the final round, \Abelard takes $(\str{N},m)$ along $\Delta\xrightarrow{\pos{\lambda}}\Delta$ to obtain $(\str{N},m+1)$ while \Heloise cannot match this move.
\end{enumerate}
\end{example}

\begin{example}\label{ex:infinite-cycles}
Let $\Delta$ be a signature consisting of one binary relation $\lambda$.
Let $\str{M}$ be the model depicted to the left of the following diagram, comprising a countably infinite number of cycles of increasing length that traverse through $0$.
Let $\str{N}$ be the model depicted to the right of the following diagram, which is obtained from $\str{M}$ by adding a countably infinite chain with no end points passing through $0$. 
Then $(\str{M},0)$ and $(\str{N},0)$ are elementarily equivalent w.\,r.\,t.\ any fragment $\frag$ of $\HPL$ closed under Boolean connectives, 
but they are not elementarily equivalent in any fragment of $\HDPL$ that is closed under Boolean connectives and has possibility over structured actions and store.
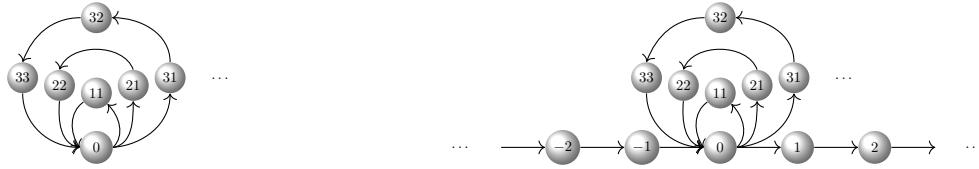
\begin{equation*}
\begin{tikzpicture}[baseline=0pt, sibling distance = 1cm, level distance=1.2cm, edge from parent/.style={draw,-latex},scale=.55, transform shape]
\begin{scope}
\node(R)[circle,ball color=white,]{$~0~$};
\node(11)[circle,ball color=white,above=15pt of R]{$11$};
\draw [->](R.east) to[out=45, in = -45] (11);
\draw [->](11) to[out=-145, in = 125] (R.west);
\node(21)[circle,ball color=white,above right = -10pt and 10pt of 11]{$21$};
\node(22)[circle,ball color=white,above left= -10pt and 10pt of 11]{$22$};
\draw [->](R.east) to[out=10, in = -90] (21.south);
\draw [->](22.south) to[out=-95, in = 170] (R.west);
\draw [->](21.north) to[out=110, in = 70] (22.north);
\node(31)[circle,ball color=white,above right = -10pt and 10pt of 21]{$31$};
\node(32)[circle,ball color=white,above = 30pt of 11]{$32$};
\node(33)[circle,ball color=white,above left = -10pt and 10pt of 22]{$33$};
\draw [->](R.east) to[out=5, in = -90] (31.south);
\draw [->](31.north) to[out=90, in = 0] (32.east);
\draw [->](32.west) to[out=180, in = 70] (33.north);
\draw [->](33.south) to[out=-90, in = 180] (R.west);
\node[ right= 14pt of 31]{$\dots$};
\end{scope}
\begin{scope}[xshift=15cm]
\node(R)[circle,ball color=white,]{$~0~$};
\node(11)[circle,ball color=white,above=15pt of R]{$11$};
\draw [->](R.east) to[out=45, in = -45] (11);
\draw [->](11) to[out=-145, in = 125] (R.west);
\node(21)[circle,ball color=white,above right = -10pt and 10pt of 11]{$21$};
\node(22)[circle,ball color=white,above left= -10pt and 10pt of 11]{$22$};
\draw [->](R.east) to[out=10, in = -90] (21.south);
\draw [->](22.south) to[out=-95, in = 170] (R.west);
\draw [->](21.north) to[out=110, in = 70] (22.north);
\node(31)[circle,ball color=white,above right = -10pt and 10pt of 21]{$31$};
\node(32)[circle,ball color=white,above = 30pt of 11]{$32$};
\node(33)[circle,ball color=white,above left = -10pt and 10pt of 22]{$33$};
\draw [->](R.east) to[out=5, in = -90] (31.south);
\draw [->](31.north) to[out=90, in = 0] (32.east);
\draw [->](32.west) to[out=180, in = 70] (33.north);
\draw [->](33.south) to[out=-90, in = 180] (R.west);
\node[ right= 14pt of 31]{$\dots$};
\node(1)[circle,ball color=white,right = 30pt of R]{$~1~$};
\node(2)[circle,ball color=white,right = 30pt of 1]{$~2~$};
\node(3)[right = 30pt of 2]{};
\node[right = 10pt of 3]{$\dots$};
\draw [->](R) to (1);
\draw [->](1) to (2);
\draw [->](2) to (3);
\node(-1)[circle,ball color=white,left = 30pt of R]{$-1$};
\node(-2)[circle,ball color=white,left = 30pt of -1]{$-2$};
\node(-3)[left = 30pt of -2]{};
\node[left = 10pt of -3]{$\dots$};
\draw [->](-1) to (R);
\draw [->](-2) to (-1);
\draw [->](-3) to (-2);
\end{scope}
\end{tikzpicture}
\end{equation*}
\Abelard has a winning strategy in three steps:
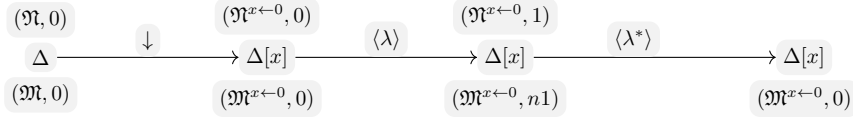
\begin{equation*}
\begin{tikzpicture}[baseline=0pt, scale=.8, transform shape]
\node[model, above=2pt of A0]{$(\str{N},0)$};
\node(A0)[sig]{$\Delta$};
\node[model, below=2pt of A0]{$(\str{M},0)$};
\node(A1)[sig, right=3cm of A0]{$\Delta[x]$};
\node[model, above=2pt of A1]{$(\str{N}\ext{x}{0},0)$};
\draw[->] (A0) to node[sig,above]{$\downarrow$} (A1) ;
\node[model, below=2pt of A1]{$(\str{M}\ext{x}{0},0)$};
\node(A2)[sig, right=3cm of A1]{$\Delta[x]$};
\node[model, above=2pt of A2]{$(\str{N}\ext{x}{0},1)$};
\draw[->] (A1) to node[sig,above]{$\pos{\lambda}$} (A2);
\node[model, below=2pt of A2]{$(\str{M}\ext{x}{0},n1)$};
\node(A3)[sig, right=4cm of A2]{$\Delta[x]$};
\draw[->] (A2) to node[sig,above,xshift=-10pt]{$\pos{\lambda^*}$} (A3);
\node[model, below=2pt of A3]{$(\str{M}\ext{x}{0},0)$};
\end{tikzpicture}
\end{equation*}
The game starts by naming the current state by both \Abelard and \Heloise.
Then \Abelard moves $(\str{N}\ext{x}{0},0)$ along $\Delta\stackrel{\pos{\lambda}}\longrightarrow\Delta$ to obtain $(\str{N}\ext{x}{0},1)$.
\Heloise moves along $\Delta\stackrel{\pos{\lambda}}\longrightarrow\Delta$ to obtain $(\str{M}\ext{x}{0},n1)$.
For the final round, \Abelard can return $(\str{M}\ext{x}{0},n1)$ to the state $0$ while \Heloise cannot match this move.
\end{example}

\subsection{Game sentences}

We propose a notion of game sentence defined over a gameboard tree which describes precisely the EF games played on the gameboard tree given.

\begin{definition}[Game sentence]
The set of game sentences $\Theta_\tr$ over a gameboard tree $\tr$, whose root is labelled by a finite signature $\Delta$, is defined by structural induction on gameboard trees:
\begin{description}
\item[$\tr = \Delta$:] 
Let $\Theta_\Delta = \{ \bigland_{\rho\in \Sen_b(\Delta)} \rho^{f(\rho)} \mid f: \Sen_b(\Delta) \to \{0,1\}\}$, where $\rho^0 = \rho$
and $\rho^1 = \neg\rho$.

\item[$\tr = \Delta(\xrightarrow{\lb_1} \tr_1 \dots \xrightarrow{\lb_n} \tr_n)$:] This case can be depicted as follows:

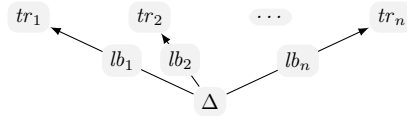
\begin{center}
\begin{tikzpicture}[baseline=0pt, sibling distance = 2cm, level distance=1.4cm, edge from parent/.style={draw,-latex},scale=.8, transform shape]
\tikzstyle{every node}=[fill=gray!10,rounded corners]
\node {$\Delta$}[grow=up]
  child{node {$\tr_n$}
        edge from parent node{$\lb_n$}
       }
  child{node {\dots}
        edge from parent[draw=none]
       }
  child{node{$\tr_2$}
        edge from parent node{$\lb_2$}  
       }
  child{node{$\tr_1$}
        edge from parent node{$\lb_1$}
       };
\end{tikzpicture}
\end{center}
Let us fix an arbitrary index $i\in\{1,\dots,n\}$.  
We will define 
\begin{enumerate}
\item a subsetet $S_i\subseteq \powerset(\Theta_{\tr_i})$ of the powerset of $\Theta_{\tr_i}$, and
\item a $\Delta$-sentence $\varphi_\Gamma$ for each set of game sentences $\Gamma\in S_i$.
\end{enumerate}
Depending on the label $\lb_i$, there are five cases:
\begin{description}
\item[$\Delta \xrightarrow{\pos{\act}} \Delta$:]
In this case, we assume that $\Diamond\in \sop$.
Then we define:
\begin{enumerate}
\item $S_i\coloneqq \powerset(\Theta_{\tr_i})$, the powerset of $\Theta_{\tr_i}$, and
\item $\varphi_\Gamma\coloneqq (\bigland_{\gamma\in\Gamma} \pos{\act}\gamma) \wedge(\nec{\act} \biglor\Gamma)$ for all $\Gamma\in S_i$.
\end{enumerate}
\item[$\Delta \xrightarrow{\at{k}} \Delta$:] 
In this case, we assume that $@\in\sop$. 
Then we define:
\begin{enumerate}
\item $S_i\coloneqq \{ \{ \gamma \} \mid \gamma\in\Theta_{\tr_i} \}$, the set of all singletons with elements from $\Theta_{\tr_i}$, and
\item $\varphi_{\{\gamma\}}\coloneqq \at{k}\gamma$ for all singletons $\{\gamma\}\in S_i$.
\end{enumerate}
\item[$\Delta \xrightarrow{\downarrow} \Delta{[x]}$:] 
In this case, we assume that $\downarrow\ \in\sop$.
The we define:
\begin{enumerate}
\item $S_i\coloneqq\{ \{ \gamma \} \mid \gamma\in\Theta_{\tr_i} \}$, the set of all singletons with elements from $\Theta_{\tr_i}$, and
\item $\varphi_{\{\gamma\}}= \store{x}\gamma$ for all singletons $\{\gamma\}\in S_i$.
\end{enumerate}
\item[$\Delta \xrightarrow{\exists} \Delta{[x]}$:]
In this case, we assume that $\exists\in \sop$.
Then we define:
\begin{enumerate}
\item $S_i\coloneqq \powerset(\Theta_{\tr_i})$, the powerset of $\Theta_{\tr_i}$, and
\item $\varphi_\Gamma\coloneqq(\bigland_{\gamma\in\Gamma} \Exists{x}\gamma) \wedge(\Forall{x} \biglor \Gamma)$,
for all $\Gamma\in S_i$.
\end{enumerate}
\item[$\Delta \xrightarrow{1} \Delta$:]
Then we define:
\begin{enumerate}
\item $S_i\coloneqq \{ \{ \gamma \} \mid \gamma\in\Theta_{\tr_i} \}$, the set of all singletons with elements from $\Theta_{\tr_i}$, and
\item $\varphi_{\{\gamma\}}\coloneqq \gamma$ for all singletons $\{\gamma\}\in S_i$.
\end{enumerate}
\end{description}

The set of game sentences over $\tr$ is
$\Theta_{\tr}= \{ \varphi_{\Gamma_1} \wedge \dots\wedge \varphi_{\Gamma_n} \mid \Gamma_1\in S_1, \dots, \Gamma_n\in
S_n\}$.
\end{description}
\end{definition}

\begin{example}
Assume that $\frag$ is $\HPL$.
Recall the signature $\Delta$, along with the $\Delta$-models $\str{M}$ and $\str{N}$ introduced in \cref{ex:loop}.
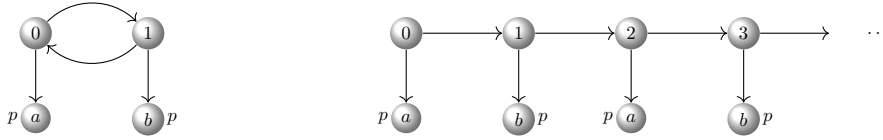
\begin{equation*}
\begin{tikzpicture}[baseline=0pt, sibling distance = 1cm, level distance=1.2cm, edge from parent/.style={draw,-latex},scale=.7, transform shape]
\begin{scope}
\node(C0)[circle,ball color=white,fill=gray!10]{$0$};
\node(C1)[circle,ball color=white,fill=gray!10,right=1.5cm of C0]{$1$};
\node(A)[circle,ball color=white,fill=gray!10,below=1cm of C0]{$a$};
\node(B)[circle,ball color=white,fill=gray!10,below=1cm of C1]{$b$};
\draw [->](C0) to[out=45, in = 135] (C1);
\draw [->](C1) to[out=-135, in = -45] (C0);
\draw [->](C0) to (A);
\draw [->](C1) to (B);
\node[left=-2pt of A]{$p$};
\node[right=-2pt of B]{$p$};
\end{scope}
\begin{scope}[xshift=7cm]
\node(C0)[circle,ball color=white,fill=gray!10]{$0$};
\node(C1)[circle,ball color=white,fill=gray!10,right=1.5cm of C0]{$1$};
\node(A0)[circle,ball color=white,fill=gray!10,below=1cm of C0]{$a$};
\node(B1)[circle,ball color=white,fill=gray!10,below=1cm of C1]{$b$};
\node(C2)[circle,ball color=white,fill=gray!10,right=1.5cm of C1]{$2$};
\node(C3)[circle,ball color=white,fill=gray!10,right=1.5cm of C2]{$3$};
\node(A2)[circle,ball color=white,fill=gray!10,below=1cm of C2]{$a$};
\node(B3)[circle,ball color=white,fill=gray!10,below=1cm of C3]{$b$};
\draw [->](C0) to (C1);
\draw [->](C0) to (A0);
\draw [->](C1) to (B1);
\draw [->](C1) to (C2);
\draw [->](C2) to (C3);
\draw [->](C2) to (A2);
\draw [->](C3) to (B3);
\node[left=-2pt of A0]{$p$};
\node[right=-2pt of B1]{$p$};
\node[left=-2pt of A2]{$p$};
\node[right=-2pt of B3]{$p$};
\node(C4)[right=1.3cm of C3]{};
\draw [->](C3) to (C4);
\node[right=10pt of C4]{$\dots$};
\end{scope}
\end{tikzpicture}
\end{equation*}
Let $\tr_0$ denote the gameboard tree shown in \cref{ex:loop}.
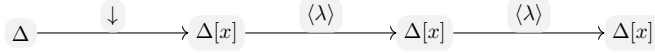
\begin{equation*}
\begin{tikzpicture}[baseline=0pt, scale=.8, transform shape]
\node(A0)[sig]{$\Delta$};
\node(A1)[sig, right=2.5cm of A0]{$\Delta[x]$};
\draw[->] (A0) to node[sig,above]{$\downarrow$} (A1) ;
\node(A2)[sig, right=2.5cm of A1]{$\Delta[x]$};
\draw[->] (A1) to node[sig,above]{$\pos{\lambda}$} (A2);
\node(A3)[sig, right=2.5cm of A2]{$\Delta[x]$};
\draw[->] (A2) to node[sig,above]{$\pos{\lambda}$} (A3);
\end{tikzpicture}
\end{equation*}
\begin{itemize}
\item Let $\tr_{i+1}$ be the immediate subtree of $\tr_i$ for all $i\in\{0,1,2\}$.
\item Note that $\Theta_{\tr_3}=\{x\land p, x\land \neg p, \neg x\land p, \neg x\land \neg p \}$.
\item Let $\Gamma_3\coloneqq\{\neg x\land p, \neg x\land \neg p\}$. 
\item Note that $\Gamma_3\subseteq \Theta_{\tr_3}$ and 
$\varphi_{\Gamma_3}=((\pos{\lambda}\neg x\land p) \land \pos{\lambda} \neg x \land \neg p)\land(\nec{\lambda}\neg x\land p \lor \neg x \land \neg p)\in \Theta_{\tr_2}$. 
\item Let $\Gamma_2\coloneqq \{\varphi_{\Gamma_3}\}$ and note that 
$\Gamma_2\subseteq \Theta_{\tr_2}$ and 
$\varphi_{\Gamma_2}=(\pos{\lambda}\varphi_{\Gamma_3})\land(\nec{\lambda}\varphi_{\Gamma_3})\in \Theta_{\tr_1}$.
\item Let $\Gamma_1\coloneqq \{\varphi_{\Gamma_2}\}$ and note that 
$\Gamma_1\subseteq \Theta_{\tr_1}$ and 
$\varphi_{\Gamma_1}=\store{x} \varphi_{\Gamma_2} \in \Theta_{\tr_0}$.
\end{itemize}
The diagram below illustrates a winning strategy for \Abelard in the EF game played on $\tr_0$.
\begin{equation*}
\begin{tikzpicture}[baseline=0pt, scale=.8, transform shape]
\node(A0)[sig]{$\Delta$};
\node[model, above=2pt of A0]{$(\str{N},0) \models \varphi_{\Gamma_1}$};
\node[model, below=2pt of A0]{$(\str{M},0)\not\models \varphi_{\Gamma_1}$};
\node(A1)[sig, right=3cm of A0]{$\Delta[x]$};
\node[model, above=2pt of A1]{$(\str{N}\ext{x}{0},0)\models \varphi_{\Gamma_2}$};
\draw[->] (A0) to node[sig,above]{$\downarrow$} (A1) ;
\node(B1)[model, below=2pt of A1]{$(\str{M}\ext{x}{0},0)\not\models \varphi_{\Gamma_2}$};
\node(A2)[sig, right=3cm of A1]{$\Delta[x]$};
\node[model, above=2pt of A2]{$(\str{N}\ext{x}{0},1)\models \varphi_{\Gamma_3}$};
\draw[->] (A1) to node[sig,above]{$\pos{\lambda}$} (A2);
\node[model, below=2pt of A2]{$(\str{M}\ext{x}{0},1)\not\models \varphi_{\Gamma_3}$};
\node(A3)[sig, right=5cm of A2]{$\Delta[x]$};
\node[model, above=2pt of A3]{$(\str{N}\ext{x}{0},2) \models \neg x\land \neg p$};
\node[model, above=20pt of A3]{$(\str{N}\ext{x}{0},b) \models \neg x\land p$~~~};
\draw[->] (A2) to node[sig,above]{$\pos{\lambda}$} (A3);
\node[model, below=2pt of A3]{$(\str{M}\ext{x}{0},0) \not \models \neg x\land \neg p$};
\node[model, below=20pt of A3]{$(\str{M}\ext{x}{0},b) \models \neg x\land p$~~~};
\end{tikzpicture}
\end{equation*}
Since $(\str{M}\ext{x}{0},0) \not \models \neg x\land \neg p$ and $(\str{M}\ext{x}{0},b) \models \neg x\land p$, we have 
$(\str{M}\ext{x}{0},1)\not\models \varphi_{\Gamma_3}$, which implies 
$(\str{M}\ext{x}{0},0)\not\models \varphi_{\Gamma_2}$;
hence $(\str{M},0)\not\models \varphi_{\Gamma_1}$.
Similarly, since $(\str{N}\ext{x}{0},b) \models \neg x\land p$ and $(\str{N}\ext{x}{0},2) \models \neg x\land \neg p$, we have 
$(\str{N}\ext{x}{0},1)\models \varphi_{\Gamma_3}$, which implies 
$(\str{N}\ext{x}{0},0)\models \varphi_{\Gamma_2}$;
hence $(\str{N},0) \models \varphi_{\Gamma_1}$.
We will show that \Abelard has a winning strategy for the EF game played on $\tr_0$ (as depicted above), because $(\str{M},0)$ and $(\str{N},0)$ satisfy different game sentences in $\Theta_{\tr_0}$.
\end{example}
Finite hybrid-dynamic EF games provide an intuitive method for establishing that two pointed
models are elementarily equivalent.  The following result shows that game
sentences characterize precisely finite hybrid-dynamic EF games.  In addition, any
hybrid-dynamic sentence is semantically equivalent to a game sentence.  Therefore, we
obtain a characterization theorem for the fragment $\frag$.

\begin{theorem}[Fra\"iss\'{e}-Hintikka theorem]\label{th:fht}
Let $\Delta$ be a finite signature.
\begin{enumerate}
  \item\label{it:th:fht:1} For all pointed models $(\str{M},w)$ defined over
$\Delta$, and all gameboard trees $\tr$ with $\rt(\tr) = \Delta$, there exists a
unique game sentence $\varphi\in \Theta_{\tr}$ such that $(\str{M},w)\models
\varphi$.
  \item\label{it:th:fht:2} For all pointed models $(\str{M},w)$ and $(\str{N},v)$
defined over $\Delta$ and all gameboard trees $\tr$ with $\rt(\tr) = \Delta$, the
following are equivalent:
\begin{enumerate}
  \item \label{th:fht:a} $(\str{M},w) \exwins_{\tr} (\str{N},v)$
  \item \label{th:fht:b} There exists a unique game sentence $\varphi\in\Theta_{\tr}$ such that
$(\str{M},w)\models\varphi$ and $(\str{N},v)\models \varphi$.
\end{enumerate}
\item\label{it:th:fht:3} 
For each sentence $\phi$ defined over $\Delta$, there exists a gameboard tree $\tr$ with $\rt(\tr) = \Delta$ and a
set of game sentences $\Psi_\phi\subseteq \Theta_{\tr}$ such that $\phi \lequiv
\biglor\Psi_\phi$ is a tautology.
\end{enumerate}
\end{theorem}

\begin{proof}
For \cref{it:th:fht:1}, we proceed by induction on gameboard trees $\tr$:
\begin{description}
\item[$\tr = \Delta$] Both existence and uniqueness follow directly from the definition of $\Theta_\Delta$.

\item[$\tr = \Delta(\xrightarrow{\lb_1} \tr_1 \dots \xrightarrow{\lb_n} \tr_n )$] 
For the induction step, 
we show that for all $i\in\{1,\dots,n\}$
there exists a unique set $\Gamma_{(\str{M},w)}^{\tr_i}\subseteq\Theta_{\tr_i}$
such that $(\str{M},w)\models \varphi_{\Gamma_{(\str{M},w)}^{\tr_i}}$.  
Let us fix an index $i\in\{1,\dots,n\}$.  
By induction hypothesis, 
for each pointed model $(\str{N},v)$ defined over the signature $\rt(\tr_i)$, 
there exists a unique game sentence $\varphi_{(\str{N},v)}^{\tr_i}\in\Theta_{\tr_i}$ such that 
$(\str{N},v)\models \varphi_{(\str{N},v)}^{\tr_i}$.  
Depending on the label $\lb_i$ we have five cases:
\begin{description}
\item[$\Delta \xrightarrow{\pos{\act}} \Delta$] 
Let $\Gamma_{(\str{M},w)}^{\tr_i}=\{ \varphi_{(\str{M},v)}^{\tr_i} \in \Theta_{\tr_i}
\mid w \mathrel{\act^{\str{M}}} v \}$, the set of all game sentences in
$\Theta_{\tr_i}$ satisfied by some pointed model $(\str{M},v)$ such that $w 
\mathrel{\act^{\str{M}}} v$ holds.  Notice that $\Gamma_{(\str{M},w)}^{\tr_i} \in S_i$.
Since $\varphi_{\Gamma_{(\str{M},w)}^{\tr_i}}=(\bigland_{\gamma\in
  \Gamma_{(\str{M},w)}^{\tr_i}} \pos{\act}\gamma) \wedge(\nec{\act}\biglor
\Gamma_{(\str{M},w)}^{\tr_i})$, we have that $(\str{M},w)\models
\varphi_{\Gamma_{(\str{M},w)}^{\tr_i}}$.
\item[$\Delta \xrightarrow{\at{k}} \Delta$] Let
$\Gamma_{(\str{M},w)}^{\tr_i}=\{\varphi_{(\str{M},k^{\str{M}})}^{\tr_i}\}$, where
$\varphi_{(\str{M},k^{\str{M}})}^{\tr_i}$ is the unique game sentence in
$\Theta_{\tr_i}$ satisfied by $(\str{M},k^{\str{M}})$.  By definition,
$\Gamma_{(\str{M},w)}^{\tr_i}\in S_i$.  Since
$\varphi_{\Gamma_{(\str{M},w)}^{\tr_i}}=
\at{k}\varphi_{(\str{M},k^{\str{M}})}^{\tr_i}$, we have that $(\str{M},w)\models
\varphi_{\Gamma_{(\str{M},w)}^{\tr_i}}$.
\item[$\Delta \xrightarrow{\downarrow} \Delta{[x]}$] 
Let $\Gamma_{(\str{M},w)}^{\tr_i}=\{\varphi_{(\str{M}\ext{x}{w},w)}^{\tr_i}\}$, where
$\varphi_{(\str{M}\ext{x}{w},w)}^{\tr_i}$ is the unique game sentence in
$\Theta_{\tr_i}$ satisfied by the pointed model $(\str{M}\ext{x}{w},w)$.  By
definition, we have $\Gamma_{(\str{M},w)}^{\tr_i} \in S_i$.  Since
$\varphi_{\Gamma_{(\str{M},w)}^{\tr_i}} =
\store{x}\varphi_{(\str{M}\ext{x}{w},w)}^{\tr_i}$, we have that
$(\str{M},w)\models \varphi_{\Gamma_{(\str{M},w)}^{\tr_i}}$.

\item[$\Delta \xrightarrow{\exists} \Delta{[x]}$] 
Let  $\Gamma_{(\str{M},w)}^{\tr_i}=\{ \varphi_{(\str{M}\ext{x}{v},w)}^{\tr_i} \in \Theta_{\tr_i} \mid v\in|\str{M}| \}$, 
the set of all game sentences in
$\Theta_{\tr_i}$ which are satisfied by some $\Delta[x]$-expansion
$\str{M}\ext{x}{v}$ of $\str{M}$ in the state $w$.  By definition, we have
$\Gamma_{(\str{M},w)}^{\tr_i} \in S_i$.  
Since
$\varphi_{\Gamma_{(\str{M},w)}^{\tr_i}}=(\bigland_{\gamma\in \Gamma_{(\str{M},w)}^{\tr_i}} \Exists{x}\gamma) \wedge(\Forall{x} \biglor \Gamma_{(\str{M},w)}^{\tr_i})$, 
we obtain  $(\str{M},w)\models \varphi_{\Gamma_{(\str{M},w)}^{\tr_i}}$.

\item[$\Delta \xrightarrow{1} \Delta$] 
Let $\Gamma^{\tr_i}_{(\str{M},w)}=\{\varphi^{\tr_i}_{(\str{M},w)}\}$ and $\varphi_{\Gamma^{\tr_i}_{(\str{M},w)}} =\varphi^{\tr_i}_{(\str{M},w)}$. 
Let $\Gamma^{\tr_i}_{(\str{M},w)}=\{\varphi^{\tr_i}_{(\str{M},w)}\}$, and by definition $\{\varphi^{\tr_i}_{(\str{M},w)}\}\in S_i$. 
By induction hypothesis, $(\str{M},w)\models \varphi^{\tr_i}_{(\str{M},w)}$.
\end{description}
Let $\varphi_{(\str{M},w)}^{\tr}=
\varphi_{\Gamma_{(\str{M},w)}^{\tr_1}}\wedge\dots\wedge\varphi_{\Gamma_{(\str{M},w)}^{\tr_n}}$.
Since $(\str{M}, w) \models \varphi_{\Gamma_{(\str{M},w)}^{\tr_i}}$ for all $i \in
\{ 1, \dots, n \}$, it follows that
$(\str{M},w)\models\varphi_{(\str{M},w)}^{\tr}$, which completes the proof of existence.

We show that $\varphi_{(\str{M},w)}^{\tr}\in\Theta_{\tr}$ is unique.  Assume that
$(\str{M},w)\models\varphi_{\Gamma_1}\wedge \dots\wedge \varphi_{\Gamma_n} $
where $\varphi_{\Gamma_1}\wedge \dots\wedge \varphi_{\Gamma_n} \in \Theta_{\tr}$.
It suffices to prove that $\Gamma_i=\Gamma_{(\str{M},w)}^{\tr_i}$ for all
$i\in\{1,\dots,n\}$.  We have four cases to consider, one for each type of
label.
\begin{description}
\item[$\Delta \xrightarrow{\pos{\act }} \Delta$] 
We prove the equality of the sets by double inclusion.
\begin{description}
  \item [$\Gamma_i\subseteq\Gamma_{(\str{M},w)}^{\tr_i}$] Since
$(\str{M},w)\models\varphi_{\Gamma_i}$, we have that
$(\str{M},w)\models\bigland_{\gamma\in\Gamma_i} \pos{\act } \gamma$; since
$(\str{M},w)\models\varphi_{\Gamma_{(\str{M},w)}^{\tr_i}}$, we obtain
$(\str{M},w)\models\nec{\act } \biglor \Gamma_{(\str{M},w)}^{\tr_i}$.  It
follows that $\Gamma_i\subseteq \Gamma_{(\str{M},w)}^{\tr_i}$.

  \item [$\Gamma_{(\str{M},w)}^{\tr_i}\subseteq \Gamma_i$] Since
$(\str{M},w)\models\varphi_{\Gamma_{(\str{M},w)}^{\tr_i}}$, we have that
$(\str{M},w)\models \bigland_{\gamma\in\Gamma_{(\str{M},w)}^{\tr_i}}
\pos{\act }\gamma$; since $(\str{M},w)\models\varphi_{\Gamma_i}$, we have
$(\str{M},w)\models \nec{\act }\biglor \Gamma_i$.  It follows that
$\Gamma_{(\str{M},w)}^{\tr_i} \subseteq \Gamma_i$.
\end{description}
\item[$\Delta \xrightarrow{\at{k}} \Delta$] Assume that $\Gamma_i=\{\gamma\}$,
where $\gamma\in\Theta_{\tr_i}$.  Recall that $\varphi_{\Gamma_i}=\at{k}\gamma$.
We have $(\str{M},w)\models \varphi_{\Gamma_i}$ iff $(\str{M},w)\models
\at{k}\gamma$ iff $(\str{M},k^{\str{M}})\models \gamma$.  It follows that $\gamma$
is $\varphi_{(\str{M},k^{\str{M}})}^{\tr_i}$, the unique game sentence in
$\Theta_{\tr_i}$ satisfied by $(\str{M},k^{\str{M}})$.  By definition,
$\Gamma_i=\Gamma_{(\str{M},k^{\str{M}})}^{\tr_i}$.
\item[$\Delta \xrightarrow{\downarrow} \Delta{[x]}$] This case is similar to
the one corresponding to $\Delta \xrightarrow{\at{k}} \Delta$.

  \item[$\Delta \xrightarrow{\exists} \Delta{[x]}$] This case is similar to the
one corresponding to $\Delta \xrightarrow{\pos{\act}} \Delta$.

\item[$\Delta \xrightarrow{1} \Delta$] 
This case is a simplification of the case corresponding to $\Delta \xrightarrow{\at{k}} \Delta$.
\end{description}
\end{description}

\smallskip\noindent%
Also for \cref{it:th:fht:2}, we proceed by induction on gameboard trees $\tr$:
\begin{description}
\item[$\tr = \Delta$] 
In this case, the equivalence of statements \eqref{th:fht:a} and \eqref{th:fht:b} follows directly from the definition of $\Theta_\Delta$.

\item[$\tr = \Delta(\xrightarrow{\lb_1} \tr_1 \dots \xrightarrow{\lb_n} \tr_n )$] 
For the forward implication, 
assume that $(\str{M},w) \exwins_{\tr} (\str{N},v)$.  
We show that $\varphi_{(\str{M},w)}^{\tr}=\varphi_{(\str{N},v)}^{\tr}$, 
which is equivalent to
$\Gamma_{(\str{M},w)}^{\tr_i}=\Gamma_{(\str{N},v)}^{\tr_i}$ for all $i\in\{1,\dots,n\}$. 
Let us fix an index $i\in\{1,\dots,n\}$.  
Depending on the label $\lb_i$, we have five cases:
\begin{description}
\item[$\Delta \xrightarrow{\pos{\act }} \Delta$] 
We show the equality by double inclusion.
\begin{description}
  \item[$\Gamma_{(\str{M},w)}^{\tr_i}\subseteq \Gamma_{(\str{N},v)}^{\tr_i}$] Let
$\gamma \in \Gamma_{(\str{M},w)}^{\tr_i}$.  By the definition of
$\Gamma_{(\str{M},w)}^{\tr_i}$, there exists a state $w_i \in
\act ^{\str{M}}(w)$ such that $(\str{M}, w_i) \models \gamma$.  Since
$(\str{M},w) \exwins_{\tr} (\str{N},v)$, there exists a state $v_i \in
\act ^{\str{N}}(v)$ such that $(\str{M},w_i) \exwins_{\tr_i} (\str{N},v_i)$. By
induction hypothesis, $(\str{M}, w_i)\models \gamma$ and $(\str{N}, v_i) \models
\gamma$.  By the definition of $\Gamma_{(\str{N},v)}^{\tr_i}$ we obtain
$\gamma\in \Gamma_{(\str{N},v)}^{\tr_i}$.
\item[$\Gamma_{(\str{N},v)}^{\tr_i} \subseteq \Gamma_{(\str{M},w)}^{\tr_i}$]
This case is symmetric to the one above.
\end{description}
\item[$\Delta \xrightarrow{\at{k}} \Delta$] 
Since 
$(\str{M},w) \exwins_{\tr} (\str{N},v)$, 
we have 
$(\str{M},k^{\str{M}}) \exwins_{\tr_i} (\str{N},k^{\str{N}})$.  
By induction hypothesis, we have
$\varphi_{(\str{M},k^{\str{M}})}^{\tr_i}=\varphi_{(\str{N},k^{\str{N}})}^{\tr_i}$.
It follows that 
$\Gamma_{(\str{M},w)}^{\tr_i}=$
$\{\varphi_{(\str{M},k^{\str{M}})}^{\tr_i}\}=$
$\{\varphi_{(\str{N},k^{\str{N}})}^{\tr_i}\}=$ 
$\Gamma_{(\str{N},v)}^{\tr_i}$.
\item[$\Delta \xrightarrow{\downarrow} \Delta{[x]}$] Since $(\str{M},w)
\exwins_{\tr} (\str{N},v)$, we have $(\str{M}\ext{x}{w},w) \exwins_{\tr_i}
(\str{N}\ext{x}{v},v)$.  By induction hypothesis, we have
$\varphi_{(\str{M}\ext{x}{w},w)}^{\tr_i}=\varphi_{(\str{N}\ext{x}{v},v)}^{\tr_i}$.
It follows that $\Gamma_{(\str{M},w)}^{\tr_i}=$
$\{\varphi_{(\str{M}\ext{x}{w},w)}^{\tr_i}\}=$
$\{\varphi_{(\str{N}\ext{x}{v},v)}^{\tr_i}\}=$ $\Gamma_{(\str{N},v)}^{\tr_i}$.
\item[$\Delta \xrightarrow{\exists} \Delta{[x]}$] We show the equality by double inclusion.
\begin{description}
\item[$\Gamma_{(\str{M},w)}^{\tr_i}\subseteq \Gamma_{(\str{N},v)}^{\tr_i}$] 
Let $\gamma\in \Gamma_{(\str{M},w)}^{\tr_i}$.  
By the definition of $\Gamma_{(\str{M},w)}^{\tr_i}$ there exists a $w_i \in |\str{M}|$ such that
$(\str{M}\ext{x}{w_i}, w) \models \gamma$.  Since $(\str{M}, w) \exwins_{\tr}
(\str{N}, v)$, there exists a $v_i\in|\str{N}|$ such that $(\str{M}\ext{x}{w_i},w)
\exwins_{\tr_i} (\str{N}\ext{x}{v_i},v)$.  
By the induction hypothesis, $(\str{M}\ext{x}{w_i}, w) \models \gamma$ and $(\str{N}\ext{x}{v_i}, v) \models \gamma$.  
By the definition of $\Gamma_{(\str{N},v)}^{\tr_i}$, we obtain $\gamma \in \Gamma_{(\str{N},v)}^{\tr_i}$

\item[$\Gamma_{(\str{N},v)}^{\tr_i} \subseteq \Gamma_{(\str{M},w)}^{\tr_i}$]
This case is symmetric to the one above.
\end{description}
\item[$\Delta \xrightarrow{1} \Delta$] 
Since 
$(\str{M},w) \exwins_{\tr} (\str{N},v)$, 
we have 
$(\str{M},w) \exwins_{\tr_i} (\str{N},v)$.  
By induction hypothesis, we have
$\varphi_{(\str{M},w)}^{\tr_i}=\varphi_{(\str{N},v)}^{\tr_i}$.
It follows that 
$\Gamma_{(\str{M},w)}^{\tr_i}=$
$\{\varphi_{(\str{M},w)}^{\tr_i}\}=$
$\{\varphi_{(\str{N},v)}^{\tr_i}\}=$ 
$\Gamma_{(\str{N},v)}^{\tr_i}$.
\end{description}

\smallskip\noindent%
For the backward implication, assume a unique $\varphi_{\Gamma_1} \wedge
\dots\wedge \varphi_{\Gamma_n}\in\Theta_{\tr}$ such that $(\str{M},w)\models
\varphi_{\Gamma_1} \wedge \dots\wedge \varphi_{\Gamma_n}$ and
$(\str{N},v)\models\varphi_{\Gamma_1} \wedge \dots\wedge \varphi_{\Gamma_n}$.  We
show that $(\str{M},w) \exwins_{\tr}(\str{N},v)$.  Assume that \Abelard moves
along an edge labeled $\lb_i$ of $\tr$.  Depending on the label $\lb_i$, we
have five cases:
\begin{description}
\item[$\Delta \xrightarrow{\pos{\act }} \Delta$] Assume that \Abelard has
chosen $w_i \in \act ^{\str{M}}(w)$.  Since $(\str{M}, w)
\models\varphi_{\Gamma_i}$, we have $(\str{M}, w) \models
\nec{\act }\biglor\Gamma_i$.  By semantics, $(\str{M}, w_i) \models \gamma_i$
for some $\gamma_i \in \Gamma_i$.  Since $(\str{N}, v)
\models\varphi_{\Gamma_i}$, we have $(\str{N}, v)
\models\bigland_{\psi\in\Gamma_i} \pos{\act }\psi$.  Since $\gamma_i\in
\Gamma_i$, we obtain $(\str{N},v)\models\pos{\act } \gamma_i$.  By semantics,
there exists $v_i\in\act ^{\str{N}}(v)$ such that $(\str{N}, v_i) \models
\gamma_i$.  Since $(\str{M},w_i)\models \gamma_i$ and $(\str{N}, v_i) \models
\gamma_i$, by induction hypothesis, $(\str{M},w_i) \exwins_{\tr_i} (\str{N},v_i)$.
\item[$\Delta \xrightarrow{\at{k}} \Delta$] 
In this case, 
$\Gamma_i = \{ \gamma_i \}$ 
for some $\gamma_i \in \Theta_{\tr_i}$ and $\varphi_{\Gamma_i} = \at{k}\gamma_i$.  
We have that $(\str{M}, w) \models \at{k}\gamma_i$, 
which means
$(\str{M}, k^{\str{M}}) \models \gamma_i$.  
Similarly, 
$(\str{N}, v)\models \at{k}\gamma_i$, 
which means $(\str{N}, k^{\str{N}})\models \gamma_i$.  
Since
$(\str{M}, k^{\str{M}}) \models \gamma_i$ and $(\str{N}, k^{\str{N}}) \models \gamma_i$, 
by induction hypothesis, 
$(\str{M}, k^{\str{M}}) \exwins_{\tr_i} (\str{N}, k^{\str{N}})$.
\item[$\Delta \xrightarrow{\downarrow} \Delta{[x]}$] In this case, $\Gamma_i =
\{ \gamma_i \}$ for some $\gamma_i \in \Theta_{\tr_i}$ and $\varphi_{\Gamma_i} =
\store{x}\gamma_i$.  We have that $(\str{M}, w)\models \store{x}\gamma_i$, which
means $(\str{M}\ext{x}{w}, w) \models \gamma_i$.  Similarly, $(\str{N}, v) \models
\store{x}\gamma_i$, which means $(\str{N}\ext{x}{v}, v) \models \gamma_i$.  Since
$(\str{M}\ext{x}{w}, w) \models \gamma_i$ and $(\str{N}\ext{x}{v},v) \models
\gamma_i$, by induction hypothesis, $(\str{M}\ext{x}{w}, w) \exwins_{\tr_i}
(\str{N}\ext{x}{v}, v)$.
\item[$\Delta \xrightarrow{\exists} \Delta{[x]}$] 
Assume that \Abelard has chosen $w_i\in|\str{M}|$.  
Since $(\str{M},w)\models\varphi_{\Gamma_i}$, 
we have $(\str{M},w)\models \Forall{x}\biglor\Gamma_i$.  
By semantics, $(\str{M}\ext{x}{w_i}, w)\models\gamma_i$ for some $\gamma_i\in\Gamma_i$.  
Since $(\str{N},v)\models\varphi_{\Gamma_i}$, we have
$(\str{N},v)\models\bigland_{\psi\in\Gamma_i} \Exists{x}\psi$.  
Since $\gamma_i\in \Gamma_i$, we obtain $(\str{N},v)\models\Exists{x}\gamma_i$.  
By semantics, $(\str{N}\ext{x}{v_i}, v)\models \gamma_i$ for some $v_i\in|\str{N}|$.
Since $(\str{M}\ext{x}{w_i},w)\models\gamma_i$ and $(\str{N}\ext{x}{v_i},v)\models
\gamma_i$, by induction hypothesis, $(\str{M}\ext{x}{w_i},w) \exwins_{\tr_i}
(\str{N}\ext{x}{v_i},v)$.

\item[$\Delta \xrightarrow{1} \Delta$]
In this case, 
$\Gamma_i = \{ \gamma_i \}$
for some $\gamma_i \in \Theta_{\tr_i}$ and $\varphi_{\Gamma_i} = \gamma_i$.  
We have that $(\str{M}, w) \models \gamma_i$.
Similarly, 
$(\str{N}, v)\models \gamma_i$.
Since
$(\str{M}, w) \models \gamma_i$ and $(\str{N}, v) \models \gamma_i$, 
by induction hypothesis, 
$(\str{M}, w) \exwins_{\tr_i} (\str{N}, v)$.
\end{description}

It follows that $(\str{M},w) \exwins_{\tr}(\str{N},v)$.
\end{description}

\smallskip\noindent%
For~\cref{it:th:fht:3},
we proceed by induction on the structure of $\phi$:
\begin{description}
  \item[$\phi\in \Sen_b(\Delta)$] Let $\Psi_{\phi}$ be the set of game sentences in
$\Theta_\Delta$ in which $\phi$ occurs positively.  It is straightforward to see
that $\biglor\Psi_\phi \lequiv \phi$ is a tautology.

  \item[$\neg\phi$] By induction hypothesis, $\phi \lequiv \biglor \Psi_\phi$
for some gameboard tree $\tr$ and some $\Psi_\phi \subseteq \Theta_\tr$.
Define the following set of game sentences over $\tr$: $\Psi_{\neg\phi} =
\Theta_\tr\setminus \Psi_\phi$.  We show that $\neg\phi\lequiv \biglor
\Psi_{\neg\phi}$ is a tautology: %

$(\str{M},w)\models\neg\phi\iff$ 
(by semantics) %

$(\str{M},w)\not\models\phi\iff$ 
(since $\phi\lequiv \biglor \Psi_\phi$ is a tautology) %

$(\str{M},w)\not\models\biglor\Psi_\phi\iff$  
(by \cref{it:th:fht:1}, which asserts that $(\str{M}, w) \models \gamma$ for some unique $\gamma \in \Theta_\tr$) %

$(\str{M},w)\models\gamma$ for some $\gamma\in \Theta_\tr\setminus \Psi_\phi\iff$
(since $\Psi_{\neg\phi}=\Theta_\tr\setminus \Psi_\phi$)

$(\str{M},w)\models\biglor \Psi_{\neg\phi}$. %

Hence, $(\str{M},w)\models\neg\phi \lequiv \biglor \Psi_{\neg\phi}$.

\item[$\bigland\Phi$]
Assume that $\Phi=\{\phi_1,\dots,\phi_n\}$. 
By induction hypothesis, for all indexes $i\in\{1,\dots,n\}$, 
$\phi_i\lequiv \biglor \Psi_{\phi_i}$ is a tautology for some gameboard tree $\tr_i$ and some set of game sentences $\Psi_{\phi_i}\subseteq \Theta_{\tr_i}$.
We construct a new gameboard tree $\tr=\Delta(\xrightarrow{1} \tr_1, \dots, \xrightarrow{1}\tr_n)$.
\begin{equation*}
\begin{tikzpicture}[baseline=0pt, sibling distance = 2cm, level distance=2cm, edge from parent/.style={draw,-latex},scale=.85, transform shape]
 \tikzstyle{every node}=[fill=gray!10,rounded corners]
\node {$\Delta$}[grow=up]
  child{node(B){$\tr_n$}
        edge from parent node{$1$}
       }
  child{node {\dots}
        edge from parent[draw=none]
        }     
  child{node(A){$\tr_1$}
        edge from parent node{$1$}
       };          
\end{tikzpicture}
\end{equation*}
Define the following set of game sentences over $\tr$:
\begin{equation*}
\Psi_{\bigland\Phi}=\{ \psi_1\land \dots\land \psi_n \mid \psi_1 \in \Psi_{\phi_1} \dots \psi_n\in \Psi_{\phi_n}\}
\end{equation*}
For all pointed models $(\str{M},w)$, we have 
$(\str{M},w)\models\bigland\Phi$ 
iff
$(\str{M},w)\models \bigland_{i=1}^n (\biglor\Psi_{\phi_i})$ 
iff
$(\str{M},w)\models \biglor \{\psi_1\land\dots\land \psi_n \mid \psi_1\in \Psi_{\phi_1},\dots,\psi_n\in \Psi_{\phi_n}\}$ 
iff 
$(\str{M},w)\models \biglor \Psi_{\bigland\Phi}$.
Hence, $(\str{M},w)\models \bigland\Phi \lequiv \biglor \Psi_{\bigland\Phi}$.
\item[$\pos{\act}\phi$] By induction hypothesis, $\phi \lequiv
\biglor\Psi_\phi$ for some gameboard tree $\tr_0$ with $\rt(\tr_0) = \Delta$ and
some set of game sentences $\Psi_\phi \subseteq \Theta_{\tr_0}$.
We construct a new gameboard tree $\tr = \Delta \xrightarrow{\pos{\act}} \tr_0$.  
We have that $(\str{M}, w) \models \pos{\act}\phi$ 
iff (since $\phi \lequiv \biglor \Psi_\phi$ is a tautology)  
$(\str{M}, w) \models \pos{\act}\biglor\Psi_\phi$ 
iff 
$(\str{M}, v) \models \biglor\Psi_\phi$ for some $v \in \act^{\str{M}}(w)$
iff (since $\Gamma_{(\str{M},w)}^{\tr_0} = \{\varphi_{(\str{M},w')}^{\tr_0} \in \Theta_{\tr_0} \mid w \mathrel{\act^{\str{M}}} w' \}$)
$\Gamma_{(\str{M},w)}^{\tr_0} \cap \Psi_\phi\neq \emptyset$.
Define the following set of game sentences over $tr$: 
$\Psi_{\pos{\act}\phi} = \{ \varphi_\Gamma \mid \Gamma \subseteq \Theta_{\tr_0} \text{ such that } \Gamma \cap \Psi_\phi \neq \emptyset \}$.
Then
$(\str{M},w)\models \pos{\act}\phi$
iff 
$\Gamma_{(\str{M},w)}^{\tr_0} \cap \Psi_\phi\neq \emptyset$ 
iff (by the definition of $\Psi_{\pos{\act}\phi}$)
$(\str{M},w)\models\biglor \Psi_{\pos{\act}\phi}$.
\item[$\at{k}\phi$] By induction hypothesis, $\phi \lequiv \biglor \Psi_\phi$
is a tautology for some gameboard tree $\tr_0$ and some set of game sentences
$\Psi_\phi\subseteq \Theta_{\tr_0}$.
Define a new tree $\tr=\Delta \xrightarrow{\at{k}} \tr_0$ and 
a the following set of game sentences over $\tr$: 
$\Psi_{(\at{k}\phi)}= \{ \at{k}\psi \mid \psi\in \Psi_\phi\}$.
Since $\phi \lequiv \biglor\Psi_\phi$ is a tautology, 
$\at{k}\phi\lequiv \biglor_{\psi\in \Psi_\phi} \at{k}\psi$ is a tautology, too.
\item[$\store{x}\phi$] By induction hypothesis, $\phi \lequiv
\biglor\Psi_\phi$ is a tautology for some gameboard tree $\tr_x$ with
$\rt(\tr_x) = \Delta[x]$ and some set of game sentences $\Psi_\phi\subseteq
\Theta_{\tr_x}$.
Define a new tree $\tr=\Delta \xrightarrow{\downarrow} \tr_x$ and 
the following set of game sentences over $\tr$:
$\Psi_{(\store{x}\phi)}=\{\store{x} \psi\mid \psi\in \Psi_\phi\}$.
Since $\phi \lequiv \biglor \Psi_\phi$ is a tautology,
$\store{x}\phi\lequiv \biglor \{\store{x} \psi\mid \psi\in \Psi_\phi\}$ is a tautology, too.
\item[$\Exists{x}\phi$] 
By induction hypothesis, $\phi\lequiv \biglor\Psi_\phi$ for some gameboard tree $\tr_x$ with $\rt(\tr_x) = \Delta[x]$
and some set of game sentence $\Psi_\phi \subseteq \Theta_{\tr_x}$.
We construct a new gameboard tree $\tr = \Delta \xrightarrow{\exists} \tr_x$.  
The following are equivalent:
$(\str{M},w)\models \Exists{x}\phi$ 
iff (since $\phi \lequiv \biglor \Psi_\phi$ is a tautology) 
$(\str{M},w)\models \Exists{x}\biglor\Psi_\phi$ 
iff
$(\str{M}\ext{x}{v},v)\models \biglor \Psi_\phi$ for some $v\in |\str{M}|$ 
iff (since $\Gamma_{(\str{M},w)}^{\tr_x} =\{ \varphi_{(\str{M}\ext{x}{w'},w)}^{\tr_x} \in \Theta_{\tr_x} \mid w'\in|\str{M}|\}$)
$\Gamma_{(\str{M},w)}^{\tr_x}\cap  \Psi_\phi\neq \emptyset$.
Define the following set of game sentences over $\tr$:
$\Psi_{\Exists{x}\phi}=\{ \varphi_\Gamma\mid \Gamma\subseteq \Theta_{\tr_x} \text{ such that } \Gamma\cap \Psi_\phi\neq\emptyset\}$.
Then $(\str{M},w)\models \Exists{x}\phi$ iff $\Gamma_{(\str{M},w)}^{\tr_x}\cap \Psi_\phi\neq \emptyset$ 
iff (by the definition of $\Psi_{\Exists{x}\phi}$)
$(\str{M},w)\models\biglor \Psi_{\Exists{x}\phi}$.
\qedhere
\end{description}
\end{proof}

\Cref{th:fht} is applicable to any fragment of $\HDPL$, including classical dynamic propositional logic for which such a result does not exist. 
The proof of \cref{th:fht} is fundamentally different from the classical case of
single-sorted first-order logic, since it relies on the construction of
gameboard trees which play the role of quantifier rank.  If the fragment $\frag$
is closed under possibility, then there is no normal form of sentences in which
first-order quantifiers are moved to the front and then store, retrieve and the
Boolean connectives are placed after.  
For this reason, the proof of the third statement requires idle moves, which serve to construct complex gameboard trees from simpler ones and to define game sentences that are semantically equivalent to conjunctions.
Notice that the above results can be straightforwardly extended to logical frameworks which allow infinitary conjunctions by constructing infinitely branched gameboard trees.

The following result is a corollary of \cref{th:fht} and it says that two
pointed models are elementarily equivalent if \Heloise has a winning strategy for
the EF games played over all gameboard trees.

\begin{corollary}\label{cor:fh}
Let $\str{M}$ and $\str{N}$ be two Kripke structures defined over a finite
signature $\Delta$.  
The following are equivalent:
\begin{enumerate}
  \item\label{cor:fh-1} $(\str{M},w)$ and $(\str{N},v)$ are $\frag$-elementarily
equivalent, in symbols $(\str{M},w) \eequiv (\str{N},v)$.

  \item\label{cor:fh-2} \Heloise has a winning strategy for the
EF game starting with $(\str{M},w)$ and $(\str{N},v)$, that is,
$(\str{M},w) \exwins_{\tr}(\str{N},v)$ for all gameboard trees $\tr$.
\end{enumerate}
\end{corollary}

\section{Countable \EF Games}\label{sec:countable-EFG}
The construction of gameboard trees of countably infinite height is straightforward.
Unlike finitary \EF games, for countably infinite \EF games, called $\omega$-EF
games for short,  
we consider only \emph{complete gameboard trees of height $\omega$}, that is, gameboard trees $\tr$ satisfying the following property:

\begin{description}
\item[{\normalfont(\textdagger)}] $\tr$ is of the following form 

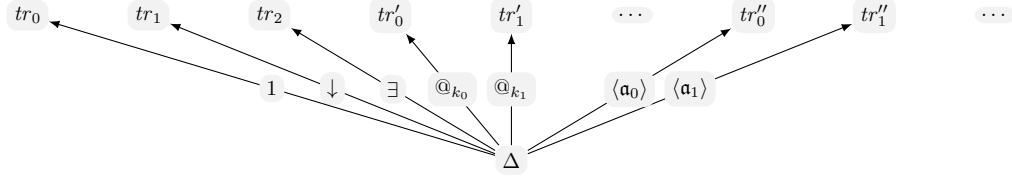
\begin{center}
\begin{tikzpicture}[baseline=0pt, sibling distance = 2cm, level distance=2.4cm, edge from parent/.style={draw,-latex},scale=.8, transform shape]
\tikzstyle{every node}=[fill=gray!10,rounded corners]
\node {$\Delta$}[grow=up]
 child{node {$\dots$}
        edge from parent[draw=none]
       }
 child{node {$\tr''_1$}
        edge from parent node{$\pos{\act_1}$}
       }
 child{node {$\tr''_0$}
        edge from parent node{$\pos{\act_0}$}
       }
  child{node {\dots}
        edge from parent[draw=none]
       }
  child{node{$\tr_1'$}
        edge from parent node{$@_{k_1}$}  
       }     
  child{node{$\tr_0'$}
        edge from parent node{$@_{k_0}$}  
       }     
  child{node{$\tr_2$}
        edge from parent node{$\exists$}  
       }     
  child{node{$\tr_1$}
        edge from parent node{$\downarrow$}  
       }
  child{node{$\tr_0$}
        edge from parent node{$1$}
       };
\end{tikzpicture}
\end{center}
where 
\begin{itemize}

\item $F=\{k_i\mid i< \alpha\}$ is an enumeration of all $\Delta$-nominals and $\alpha$ is the cardinal of $\Sen(\Delta)$,
\item $\Act(\Delta)=\{\act_i\mid  i<\alpha \}$  is an enumeration of all actions defined over $\Delta$, and
\item all $\tr_0$, $\tr_1$, $\tr_2$, $\tr_i'$ and $\tr_i''$ from the diagram above satisfy property $(\dagger)$.
\end{itemize}
\end{description}

Similarly to the finitary case, \Heloise loses the game if the game property is not satisfied by the current pair of pointed models; \Heloise wins the game if she can match any move made by \Abelard such that the game property is satisfied.  
We write $(\str{M},w)\exwins_\omega(\str{N},v)$ if \Heloise has a winning strategy for all EF games played over the complete gameboard tree of height $\omega$.

\subsection{Bisimulations and countable EF games} \label{sec:bisim}
We adapt the notion of bisimulation proposed in \cite{ArecesBM01} to our
hybrid-dynamic setting. 
Then we show that the equivalence determined by countable EF games is an
$\omega$-bisimilarity. 

\begin{definition}[$\omega$-bisimulation]\label{def:k-bisim}
Let $\str{M}$ and $\str{N}$ be Kripke structures.  
A relation $B_\ell \subseteq (|\str{M}|^\ell \times |\str{M}|) \times (|\str{N}|^\ell \times |\str{N}|)$ is an \emph{$\ell$-bisimulation} from $\str{M}$ to $\str{N}$ if for all $(\overline{w}, w) \mathrel{B_\ell} (\overline{v},v)$ the following hold:
\begin{description}
\item[(prop)]\label{it:k-bisim-prop} 
$p\in M(w)$ iff $p\in N(v)$ for all propositional symbols $p \in \Prop$;

\item[(nom)]\label{it:k-bisim-nom} $w = k^\str{M}$ iff $v = k^\str{N}$ for all nominals $k\in F$;

\item[(wvar)]\label{it:k-bisim-wvar} $\overline{w}(j) = w$ iff $\overline{v}(j) = v$ for all $1 \leq j \leq \ell$; 
\footnote{$\overline{w}(j)$ and $\overline{v}(j)$ denote the elements at position $j$ in the sequences $\overline{w}$ and $\overline{v}$, respectively.}

\item[(forth)]\label{it:k-bisim-forth} 
if $\Diamond\in \sop$ then  for all actions $\act\in \Act(\Delta)$ and all states $w' \in \act^\str{M}(w)$ there exists $v' \in \act^\str{N}(v)$ such that $(\overline{w}, w') \mathrel{B_\ell} (\overline{v}, v')$;

\item[(back)]\label{it:k-bisim-back} 
if $\Diamond\in \sop$ then for all actions $\act\in \Act(\Delta)$ and all states $v' \in \act^\str{N}(v)$ there exists $w' \in \act^\str{M}(w)$ such that $(\overline{w}, w') \mathrel{B_\ell} (\overline{v}, v')$;

\item[(atv)]\label{it:k-bisim-atv} 
if $@\in \sop$ then $(\overline{w}, \overline{w}(j)) \mathrel{B_\ell} (\overline{v}, \overline{v}(j))$ for all $1 \leq j \leq \ell$;
\footnote{The elements of $\overline{w}$ and $\overline{v}$ can be regarded as named states. Then, in accordance with (atv), the current states $w$ and $v$ are updated to the named states $\overline{w}(j)$ and $\overline{v}(j)$, respectively.}

\item[(atn)]\label{it:k-bisim-atn} 
if $@\in \sop$ then $(\overline{w}, k^\str{M}) \mathrel{B_\ell} (\overline{v}, k^\str{N})$ for all nominals $k \in F$;
\end{description}
An \emph{$\omega$-bisimulation} from $\str{M}$ to $\str{N}$ is a family of $\ell$-bisimulations $B = (B_\ell)_{\ell\in\omega}$ from $\str{M}$ to $\str{N}$ such that for all natural numbers $\ell \in \omega$ and all tuples  $(\overline{w}, w) \in |\str{M}|^\ell \times |\str{M}|$ and  $(\overline{v}, v) \in |\str{N}|^\ell \times |\str{N}|$ the following conditions are satisfied:
\begin{description}
\item[(st)] \label{it:o-bisim-st}
if ${\downarrow}\in\sop$ and $(\overline{w}, w) \mathrel{B_\ell} (\overline{v},v)$ then $(\overline{w}~w, w) \mathrel{B_{\ell+1}} (\overline{v}~v,v)$, 

where the juxtaposition stands for the concatenation of sequences; and

\item [(ex)] \label{it:o-bisim-ex}
if $\exists\in\sop$ and $(\overline{w}, w) \mathrel{B_\ell} (\overline{v},v)$ then:
\begin{description}
\item [(ex-f)] for all $w'\in |\str{M}|$ there is $v'\in |\str{N}|$ such that $(\overline{w}~w', w) \mathrel{B_{\ell+1}} (\overline{v}~v',v)$,
\item [(ex-b)] for all $v'\in |\str{N}|$ there is $w'\in |\str{M}|$ such that $(\overline{w}~w', w) \mathrel{B_{\ell+1}} (\overline{v}~v',v)$.
\end{description}
\end{description}
Two pointed models $(\str{M}, w)$ and $(\str{N}, v)$ are \emph{$\omega$-bisimilar} if there exists an \emph{$\omega$-bisimulation} $B$ from $\str{M}$ to $\str{N}$ such that $w \mathrel{B_0} v$.
In this case, we write $(\str{M}, w)\equiv_B^\frag(\str{N}, v)$.
Note that an $\omega$-bisimulation is a relation between $|\str{M}|^*$ and $|\str{N}|^*$.
\end{definition}
When there is no danger of confusion, we drop the superscript $\frag$ from the notation $\equiv^\frag_B$. 
\Cref{def:k-bisim} is obtained from the definition of  $\omega$-bisimulation from~\cite[Sect.~3.3]{ArecesBM01} by removing the condition  (bind), which already is covered by rule (st).
\begin{lemma} \label{lemma:bis-x}
Assume an $\omega$-bisimulation $B$ between $(\str{M}, \mu)$ and $(\str{N}, \nu)$ such that $(\mu,\mu')\mathrel{B_1} (\nu,\nu')$ for some tuples $(\mu,\mu')\in |\str{M}|\times |\str{M}|$ and $(\nu,\nu')\in |\str{N}|\times |\str{N}|$.
Let be $\chi:\Delta\to\Delta[x]$ be a signature extension with a variable $x$.
Then $B^x=(B^x_\ell)_{\ell\in\omega}$ defined by
\begin{gather*}
(\overline{w},w)\mathrel{B^x_\ell}(\overline{v},v) \text{ iff } (\mu\,\overline{w},w)\mathrel{B_{\ell+1}}(\nu\,\overline{v},v)\text{,}\\ 
\qquad\text{for all } \ell\in\omega\text{, all } (\overline{w},w)\in |\str{M}|^\ell\times |\str{M}| \text{ and all } (\overline{v},v)\in |\str{N}|^\ell\times |\str{N}|\text{,}
\end{gather*}
is a bisimulation between $(\str{M}\ext{x}{\mu},\mu')$ and $(\str{N}\ext{x}{\nu},\nu')$.
\end{lemma}
\begin{proof}
First, notice that $\mu'\mathrel{B_0^x}\nu'$.
Secondly, we show that all conditions from \Cref{def:k-bisim} are satisfied.
Assume that $(\overline{w},w)\mathrel{B^x_\ell}(\overline{v},v)$,
which is equivalent to 
$(\mu\,\overline{w},w)\mathrel{B_{\ell+1}}(\nu\,\overline{v},v)$.
\begin{description}
\item[(prop)]  For all $p\in \Prop$ we have:

$p\in M\ext{x}{\mu}(w)\iff$ 
(since $M\ext{x}{\mu}(w)=M(w)$) 

$p\in M(w)\iff$ 
(since $(\mu\,\overline{w},w)\mathrel{B_{\ell+1}}(\nu\,\overline{v},v)$)

$p\in N(v)\iff$ 
(since $N\ext{x}{\nu}(v)=N(v)$)

$p\in N\ext{x}{\nu}(v)$.

\item[(nom)] For all $k\in F$ we have:

$w=k^{(\str{M}\ext{x}{\mu})}\iff$  
(since $k^{(\str{M}\ext{x}{\mu})}=k^\str{M}$ )

$w=k^\str{M}\iff$ 
(since $(\mu\,\overline{w},w)\mathrel{B_{\ell+1}}(\nu\,\overline{v},v)$)

$v=k^\str{N}\iff$ 
(since $k^\str{N}=k^{(\str{N}\ext{x}{\nu})}$)

$v=k^{(\str{N}\ext{x}{\nu})}$.

In addition we have:

$w=x^{(\str{M}\ext{x}{\mu})}\iff$ 
(since $x^{(\str{M}\ext{x}{\mu})}=\mu$)

$w=\mu \iff$ 
(by (wvar), since $(\mu\,\overline{w},w)\mathrel{B_{\ell+1}}(\nu\,\overline{v},v)$)

$v=\nu \iff$ (since $\nu=x^{(\str{N}\ext{x}{\nu})}$)

$v=x^{(\str{N}\ext{x}{\nu})}$.

\item[(wvar)]
For all $1\leq j\leq \ell$ we have:

$w=\overline{w}(j)\iff$ 
(by (wvar), since $(\mu\,\overline{w},w)\mathrel{B_{\ell+1}}(\nu\,\overline{v},v)$) 

$v=\overline{v}(j)$.

\item[(forth)]
Let $\act\in A$ be an action and let $w'\in \act^\str{M}(w)$ be a state.

Since $(\mu\,\overline{w},w)\mathrel{B_{\ell+1}}(\nu\,\overline{v},v)$, 
by (forth), $(\mu\,\overline{w},w')\mathrel{B_{\ell+1}}(\nu\,\overline{v},v')$
for some $v'\in \act^\str{N}(v)$.

By the definition of $B^x$,  we get  $(\overline{w},w')\mathrel{B_\ell^x}(\overline{v},v')$.

\item[(back)] Similar to (forth).

\item [(st)] 
By (st) for $B$, we have $(\mu\,\overline{w}\,w,w)\mathrel{B_{\ell+2}}(\nu\,\overline{v}\,v,v)$.

By the definition of $B^x$, we get $(\overline{w}\,w,w)\mathrel{B_{\ell+1}^x}(\overline{v}\,v,v)$.
 
\item[(ex-f)]
Let $w'\in |\str{M}|$ be any state from $\str{M}$.

By (ex-f) property of $B$, 
$(\mu\,\overline{w}\,w',w)\mathrel{B_{\ell+2}}(\nu\,\overline{v}\,v',v)$ for some $v'\in |\str{N}|$.

By the definition of $B^x$, we get $(\overline{w}\,w',w)\mathrel{B_{\ell+1}^x}(\overline{v}\,v',v)$.

\item[(ex-b)]
Similar to (ex-f).\qedhere
\end{description}
\end{proof}

Notice that if ${\downarrow}\in\sop$ or $\exists\in\sop$, then $B^x$ defined in
the lemma above is not empty.

\begin{theorem}[Bisimulations vs. \EF games] \label{th:bis-ef}
Let $(\str{M},w)$ and $(\str{N},v)$ be two pointed models defined over a signature $\Delta$.
Then:
\begin{equation*}
(\str{M},w)\equiv_B(\str{N},v) \quad\text{iff}\quad (\str{M},w)\approx_\omega(\str{N},v)
\text{.}
\end{equation*}
\end{theorem}
\begin{proof}
For the forward implication, assume an $\omega$-bisimulation $B$ between $(\str{M},w)$ and $(\str{N},v)$.
By (prop) and (nom), $(\str{M},w)$ and $(\str{N},v)$ satisfy the same atomic sentences.
We show that each move made by \Abelard can be matched by a move made by \Heloise such that the resulting pointed models are again bisimilar.
There are four cases to consider depending on the label of the edge along which \Abelard moves:
\begin{description}
\item[$\Delta \xrightarrow{\at{k}} \Delta$]
By (atn), $k^\str{M}\mathrel{B_0}k^\str{N}$. 
Also, $B$ is an $\omega$-bisimulation between $(\str{M},k^\str{M})$ and $(\str{N},k^\str{N})$.
\item[$\Delta \xrightarrow{\pos{\act}} \Delta$] 
Assume that \Abelard has chosen $w' \in \act^\str{M} (w)$.  
By (forth), there exists $v' \in\act^\str{N} (v)$ such that $w'\mathrel{B_0}v'$.
Since $B$ is an $\omega$-bisimulation between  $(\str{M},w')$ and $(\str{N},v')$,
\Heloise can choose $v'$.
\item[$\Delta \xrightarrow{\downarrow} \Delta{[x]}$]
The resulting pointed models are $(\str{M}\ext{x}{w},w)$ and $(\str{N}\ext{x}{v},v)$.
Since $w\mathrel{B_0}v$, 
by (st), we have $(w,w)\mathrel{B_1}(v,v)$.
By \cref{lemma:bis-x}, there exists an $\omega$-bisimulation $B^x$ between  $(\str{M}\ext{x}{w},w)$ and $(\str{N}\ext{x}{v},v)$.
\item[$\Delta \xrightarrow{\exists} \Delta{[x]}$]
Assume \Abelard has chosen $w'\in|\str{M}|$.
Since $w\mathrel{B_0}v$, by (ex-f), there exists $v'\in|\str{N}|$ such that $(w',w)\mathrel{B_1}(v',v)$.
By \cref{lemma:bis-x}, there exists an $\omega$-bisimulation $B^x$ between  $(\str{M}\ext{x}{w'},w)$ and $(\str{N}\ext{x}{v'},v)$.
Hence, \Heloise can choose $v'\in|\str{N}|$.
\end{description}

For the backward implication, 
for each $\ell\in\omega$, we define $(\overline{\mu},\mu)\mathrel{B_\ell}(\overline{\nu},\nu)$ iff there exists a sequence of moves for the \EF game as depicted in the following diagram:
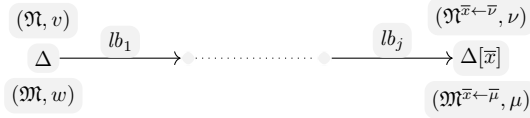
\begin{equation*}
\begin{tikzpicture}[baseline=0pt, scale=.8, transform shape]
\node[model, above=2pt of A0]{$(\str{N},v)$};
\node(A0)[sig]{$\Delta$};
\node[model, below=2pt of A0]{$(\str{M},w)$};
\node(A1)[sig, right=2cm of A0]{};
\draw[->] (A0) to node[sig,above]{$\lb_1$} (A1);
\node(A2)[sig, right=2cm of A1]{};
\draw[dotted] (A1) to (A2);
\node(A3)[sig, right=2cm of A2]{$\Delta[\overline{x}]$};
\draw[->] (A2) to node[sig,above]{$\lb_j$} (A3);
\node[model, above=2pt of A3]{$(\str{N}\ext{\overline{x}}{\overline{\nu}},\nu)$};
\node[model, below=2pt of A3]{$(\str{M}\ext{\overline{x}}{\overline{\mu}},\mu)$};
\end{tikzpicture}
\end{equation*}
Notice that 
\begin{itemize}
\item $\overline{x}$ is a sequence of $\ell$ variables and $\overline{x}(j)$ is the element at position $j$ in $\overline{x}$ for all $j\in\{1,\dots,\ell\}$;
\item $\str{M}\ext{\overline{x}}{\overline{\mu}}$ is the unique expansion of $\str{M}$ to the signature $\Delta[\overline{x}]$ interpreting each variable $\overline{x}(i)$ as the state $\overline{\mu}(i)$ for all $i\in\{1,\dots,\ell\}$;
\item $\str{N}\ext{\overline{x}}{\overline{\nu}}$ is the unique expansion of $\str{N}$ to the signature $\Delta[\overline{x}]$ interpreting each variable $\overline{x}(i)$ as the state $\overline{\nu}(i)$ for all $i\in\{1,\dots,\ell\}$; and
\item $(\str{M}\ext{\overline{x}}{\overline{\mu}},\mu) \approx_\omega (\str{N}\ext{\overline{x}}{\overline{\nu}},\nu)$.
\end{itemize}
We show that $B=(B_\ell)_{\ell\in\omega}$ is an $\omega$-bisimulation.
Assume that  $(\overline{\mu},\mu)\mathrel{B_\ell}(\overline{\nu},\nu)$.
\begin{description}
\item[(prop)] For any propositional symbol $p\in\Prop$,
since $(\str{M}\ext{\overline{x}}{\overline{\mu}},\mu) \approx_\omega (\str{N}\ext{\overline{x}}{\overline{\nu}},\nu)$,
we have
$p\in M\ext{\overline{x}}{\overline{\mu}}(\mu) \iff p\in N\ext{\overline{x}}{\overline{\nu}}(\nu)$.
Since $M\ext{\overline{x}}{\overline{\mu}}(\mu) = M(\mu)$ and $N\ext{\overline{x}}{\overline{\nu}}(\nu)=N(\nu)$,
we get $p\in M(\mu) \iff p\in N(\nu)$.
\item[(nom)]
For any nominal $k\in F$,
since $(\str{M}\ext{\overline{x}}{\overline{\mu}},\mu) \approx_\omega (\str{N}\ext{\overline{x}}{\overline{\nu}},\nu)$,
we have
$k^{(\str{M}\ext{\overline{x}}{\overline{\mu}})}=\mu \iff k^{(\str{N}\ext{\overline{x}}{\overline{\nu}})}=\nu$.
Since $k^{(\str{M}\ext{\overline{x}}{\overline{\mu}})}=k^\str{M}$ and $k^{(\str{N}\ext{\overline{x}}{\overline{\nu}})}=k^\str{N}$, 
we get $k^\str{M}=\mu\iff k^\str{N}=\nu$.
\item[(wvar)]
For all indexes $j\in\{1,\dots,\ell\}$,
since $(\str{M}\ext{\overline{x}}{\overline{\mu}},\mu) \approx_\omega (\str{N}\ext{\overline{x}}{\overline{\nu}},\nu)$,
we have
$\overline{x}(j)^{(\str{M}\ext{\overline{x}}{\overline{\mu}})}=\mu\iff \overline{x}(j)^{(\str{N}\ext{\overline{x}}{\overline{\nu}})}=\nu$.
Since  $\overline{x}(j)^{(\str{M}\ext{\overline{x}}{\overline{\mu}})}=\overline{\mu}(j)$ and $\overline{x}(j)^{(\str{N}\ext{\overline{x}}{\overline{\nu}})}=\overline{\nu}(j)$, we get 
$\overline{\mu}(j)=\mu\iff \overline{\nu}(j)=\nu$.
\item[(forth)]
Let $\act$ be an action, and let $\mu'\in\act^\str{M}(\mu)$ be a state.
Since $(\str{M}\ext{\overline{x}}{\overline{\mu}},\mu) \approx_\omega (\str{N}\ext{\overline{x}}{\overline{\nu}},\nu)$,
there exists $\nu'\in\act^\str{N}(\nu)$ such that  $(\str{M}\ext{\overline{x}}{\overline{\mu}},\mu') \approx_\omega (\str{N}\ext{\overline{x}}{\overline{\nu}},\nu')$.
By the definition of $B_\ell$, we obtain $(\overline{\mu},\mu')\mathrel{B_\ell}(\overline{\nu},\nu')$.
\item[(back)] Similar to (forth).
\item[(st)] 
Since $(\str{M}\ext{\overline{x}}{\overline{\mu}},\mu) \approx_\omega (\str{N}\ext{\overline{x}}{\overline{\nu}},\nu)$, 
we get $(\str{M}\ext{\overline{x}x}{\overline{\mu}\mu},\mu) \approx_\omega (\str{N}\ext{\overline{x}x}{\overline{v}\nu},\nu)$,
where $x$ is a variable for $\Delta[\overline{x}]$.
By the definition of $B_{\ell+1}$,
we obtain $(\overline{\mu}\,\mu,\mu) \mathrel{B_{\ell+1}} (\overline{\nu}\,\nu,\nu)$.
\item[(ex-f)]
Let $\mu'\in |\str{M}|$.
Since $(\str{M}\ext{\overline{x}}{\overline{\mu}},\mu) \approx_\omega (\str{N}\ext{\overline{x}}{\overline{\nu}},\nu)$, 
we get $(\str{M}\ext{\overline{x}x}{\overline{\mu}\mu'},\mu) \approx_\omega (\str{N}\ext{\overline{x}x}{\overline{\nu}\nu'},\nu)$ 
for some $\nu'\in|\str{N}|$,
where $x$ is a variable for $\Delta[\overline{x}]$. 
By the definition of $B_{\ell+1}$,
we obtain $(\overline{\mu}\,\mu',\mu) \mathrel{B_{\ell+1}} (\overline{\nu}\,\nu',\nu)$.\qedhere
\end{description}
\end{proof}

\subsection{Back-and-forth systems and countable EF games} \label{sec:backnforth}
Analogously to \cref{sec:bisim}, 
we adapt the notion of back-and-forth system proposed in~\cite{ArecesBM01}
to our setting, and show that under certain conditions
back-and-forth equivalence and the countable EF equivalence ($\approx_\omega$)
coincide. It turns out  however, that in general
our back-and-forth equivalence is stronger than $\omega$-EF equivalence. 
Since the conditions we identify here are satisfied in the case considered
in~\cite{ArecesBM01}, this does not contradict
Theorem~3.7 and Corollary~3.12 there.

\begin{definition}[Basic partial isomorphism]\label{def:iso}
Let $\str{M}$ and $\str{N}$ be two models over a signature $\Delta$.
A \emph{basic partial isomorphism} $h:\str{M} \nrightarrow \str{N}$ is a bijection
from a subset of $|\str{M}|$ to a subset of $|\str{N}|$ such that for all $w\in \dom(h)$ and all $\rho\in\Sen_b(\Delta)$ we have
\begin{equation*}
(\str{M},w) \models \rho \quad \text{ iff }\quad (\str{N}, h(w)) \models \rho
\ \text{.}
\end{equation*}
The basic partial isomorphism $g : \str{M} \nrightarrow \str{N}$ \emph{extends}
$h$, written $h \subseteq g$, if $\dom(h) \subseteq \dom(g)$ and $g(w) = h(w)$
for all $w \in \dom(h)$.
\end{definition}

A \emph{partial isomorphism}~\cite{ArecesBM01} $h:\str{M}\nrightarrow\str{N}$ is a
basic partial isomorphism such that for all binary relations $\lambda$ in $\Delta$  and all states $w_1, w_2\in \dom(h)$ we have
\begin{equation*}  
w_1 \mathrel{\lambda^{\str{M}}} w_2 \quad\text{ iff }\quad h(w_1) \mathrel{\lambda^{\str{N}}} h(w_2)
\ \text{.}
\end{equation*}
Since the underlying logic $\frag$ is obtained from $\HDPL$ by dropping some of
the sentence or action constructors, the following definition is given by cases.

\begin{definition}[Back-and-forth system]\label{def:baf}
A back-and-forth system between two Kripke structures $\str{M}$ and $\str{N}$
defined over a signature $\Delta = ((F, P), \Prop)$ is a non-empty family
$\bfsys$ of basic partial isomorphisms between $\str{M}$ and $\str{N}$ satisfying
the following properties:
\begin{description}
  \item[$@$-extension] If $\frag$ is closed under retrieve, then for all $h \in
\bfsys$ and all $k \in F$, there exists a $g \in \bfsys$ such that $h
\subseteq g$ and $k^{\str{M}} \in \dom(g)$.

  \item[$\Diamond$-extension] If $\frag$ is closed under possibility over an action $\act$, then:
\begin{description}
  \item[forth] for all $h \in \bfsys$, all $w_1 \in \dom(h)$ and all $w_2
\in |\str{M}|$ such that $w_1 \mathrel{\act^{\str{M}}} w_2$, there exists a $g
\in \bfsys$ such that $h \subseteq g$, $w_2\in \dom(g)$, and $g(w_1)
\mathrel{\act^{\str{N}}} g(w_2)$;

  \item [back] for all $h \in \bfsys$, all $v_1 \in \rng(h)$ and all $v_2 \in |\str{N}|$ such that $v_1 \mathrel{\act^{\str{N}}} v_2$, there exists a $g\in \bfsys$ such that $h \subseteq g$, $v_2 \in \rng(g)$, and $g^{-1}(v_1) \mathrel{\act^{\str{M}}} g^{-1}(v_2)$.
\end{description}
 
  \item[$\exists$-extension] If $\frag$ is closed under existential quantifiers,
then:
\begin{description}
  \item[forth] for all $h \in \bfsys$ and all $w \in |\str{M}|$, there
exists a $g \in \bfsys$ such that $h \subseteq g$ and $w \in \dom(g)$;

  \item [back] for all $h \in \bfsys$ and all $v \in |\str{N}|$, there
exists a $g \in \bfsys$ such that $h \subseteq g$ and $v \in \rng(g)$.
\end{description}
\end{description}
Two Kripke structures $\str{M}$ and $\str{N}$ are \emph{back-and-forth
  equivalent}, if there is a back-and-forth system $\bfsys$ between $\str{M}$
and $\str{N}$, in symbols, $\str{M} \bfequiv_{\bfsys} \str{N}$.  Two pointed
models $(\str{M}, w)$ and $(\str{N}, v)$ are \emph{back-and-forth equivalent},
if there is a back-and-forth system $\bfsys$ between $\str{M}$ and $\str{N}$
such that $h(w) = v$ for some $h \in \bfsys$, in symbols, $(\str{M},w)
\bfequiv_{\bfsys} (\str{N},v)$.
\end{definition}

If $\frag$ is $\HPL$, then the definition of back-and-forth system proposed in this paper is equivalent with the definition of back-and-forth system proposed in \cite{ArecesBM01}:

\begin{lemma} \label{lemma:bpi-pi}
Assume that $\Diamond\in\sop$.
Any basic partial isomorphism that belongs to a back-and-forth system is a partial isomorphism.
\end{lemma}
\begin{proof}
Let $h$ be a basic partial isomorphism belonging to the back-and-forth system
$\bfsys$.  Let $w_1, w_2 \in \dom(h)$ such that $w_1 \mathrel{\lambda^{\str{M}}}
w_2$.  Then by the forth $\pos{\lambda}$-extension, there exists $g \in \bfsys$
such that $h \subseteq g$ and $g(w_1) \mathrel{\lambda^{\str{N}}} g(w_2)$.  Since
$g(w_1) = h(w_1)$ and $g(w_2) = h(w_2)$, we get $h(w_1)
\mathrel{\lambda^{\str{N}}} h(w_2)$.  For the backward implication, the arguments
are the same but we use the back $\pos{\lambda}$-extension instead of the forth
$\pos{\lambda}$-extension.
\end{proof}

In general, back-and-forth equivalence is stronger than the equivalence provided by countably infinite EF games. 
If $\frag$ satisfies certain closure conditions, back-and-forth equivalence coincides with game equivalence.
\begin{theorem}[Back-and-forth systems vs.\ \EF games] \label{th:BF-EF}
Let $(\str{M},w)$ and $(\str{N},v)$ be two pointed models defined over a signature $\Delta$.
\begin{enumerate}
\item \label{th:BF-EF-1} If $(\str{M},w)\bfequiv_{\bfsys}(\str{N},v)$ for some back-and-forth system $\bfsys$, then $(\str{M},w)\exwins_\omega(\str{N},v)$.
\item \label{th:BF-EF-2} Assume that
\begin{enumerate}
\item $\downarrow\ \in \sop$, and
\item $@\in\sop$ if $\Diamond\in \sop$ or $\exists\in \sop$.
\end{enumerate}
If $(\str{M},w)\exwins_\omega(\str{N},v)$ then $(\str{M},w)\bfequiv_{\bfsys}(\str{N},v)$ for some back-and-forth system $\bfsys$.
\end{enumerate}
\end{theorem}
\begin{proof}
We prove only the second statement. 
Assume a sequence of variables $\overline{z_n}=z_1\dots z_n$.
Assume two sequences of elements $\overline{w_n}=w_1 \dots w_n$ and $\overline{v_n}=v_1 \dots v_n$ from $\str{M}$ and $\str{N}$, respectively,  
such that 
\begin{equation*}
(\str{M}\ext{\overline{z_i}}{\overline{w_i}},\allowbreak w_{i+1}) \approx_\omega 
(\str{N}\ext{\overline{z_i}}{\overline{v_i}},\allowbreak v_{i+1})
\text{ for all }i \in \{1,\dots,{n-1}\},
\end{equation*}
where 
$\overline{z_i}=z_1\dots z_i$,
$\overline{w_i}=w_1\dots w_i$, and
$\overline{v_i}=v_1\dots v_i$.
Since $\downarrow\ \in\sop$, 
for each index $i \in \{ 1, \dots, n-1 \}$, 
a move along the edge $\Delta[\overline{z_i}] \xrightarrow{\downarrow} \Delta[\overline{z_{i+1}}]$ 
results in  $(\str{M}\ext{\overline{z_{i+1}}}{\overline{w_{i+1}}}, w_{i+1}) \approx_\omega(\str{N}\ext{\overline{z_{i+1}}}{\overline{v_{i+1}}}, v_{i+1})$.  
Let $h : \str{M} \nrightarrow \str{N}$ be the basic partial isomorphism defined by $h(w_i) = v_i$ for all $i \in \{ 1, \dots, n \}$.
It is not difficult to show that $h$ is well-defined. 
Moreover, $h$ can be extended to another basic partial isomorphism $h \cup \{ (w_{n+1}, v_{n+1}) \}$ according to \cref{def:baf} such that 
$(\str{M}\ext{\overline{z_{n+1}}}{\overline{w_{n+1}}},\allowbreak w_{n+1}) \approx_\omega 
(\str{N}\ext{\overline{z_{n+1}}}{\overline{v_{n+1}}},\allowbreak v_{n+1})$.
\begin{description}
\item [$@$-extension] In this case, $@\in\sop$.  
Let $k$ be any $\Delta$-nominal.
Let $w_{n+1} \coloneqq k^{\str{M}}$ and $v_{n+1} \coloneqq k^{\str{N}}$.
Consider a move along $\Delta[\overline{z_n}] \xrightarrow{\at{k}} \Delta[\overline{z_n}]$.  
We get
$(\str{M}\ext{\overline{z_n}}{\overline{w_n}},\allowbreak w_{n+1}) 
\approx_\omega (\str{N}\ext{\overline{ z_n}}{\overline{v_n}},\allowbreak v_{ n+1})$.
Since $\downarrow\ \in \sop$,
$(\str{M}\ext{\overline{z_{n+1}}}{\overline{w_{n+1}}},\allowbreak w_{n+1}) \approx_\omega 
(\str{N}\ext{\overline{z_{n+1}}}{\overline{v_{n+1}}},\allowbreak v_{n+1})$.
Let $g : \str{M} \nrightarrow \str{N}$ be the basic partial isomorphism $h \cup \{ (w_{n+1}, v_{n+1}) \}$.
\item [$\pos{\act}$-extension] 
In this case, $\Diamond\in\sop$ and $@\in\sop$.  
We show that there is a forth $\Diamond$-extension of $h$; 
for the back $\Diamond$-extension, the arguments are symmetric.
Assume that $w_i \mathrel{\act^{\str{M}}} w_{n+1}$ holds for some $i\in\{1,\dots,n\}$.  
Consider a move along the edge 
$\Delta[\overline{z_n}] \xrightarrow{\at{z_i}} \Delta[\overline{z_n}]$ 
to obtain 
$(\str{M}\ext{\overline{z_n}}{\overline{w_n}},w_i)$ 
and 
$(\str{N}\ext{\overline{z_n}}{\overline{v_n}},v_i)$.
Then make another move along
$\Delta[\overline{z_n}] \xrightarrow{\pos{\act}} \Delta[\overline{z_n}]$
to obtain
$(\str{M}\ext{\overline{ z_n}}{\overline{w_n}}, w_{n+1}) \approx_\omega 
(\str{N}\ext{\overline{ z_n}}{\overline{v_n}}, v_{n+1})$ 
such that 
$v_i \mathrel{\act^{\str{N}}} v_{n+1}$.
Since $\downarrow\ \in \sop$,
$(\str{M}\ext{\overline{z_{n+1}}}{\overline{w_{n+1}}},\allowbreak w_{n+1}) \approx_\omega 
(\str{N}\ext{\overline{z_{n+1}}}{\overline{v_{n+1}}},\allowbreak v_{n+1})$.
Let $g : \str{M} \nrightarrow \str{N}$ be the basic partial isomorphism $h \cup \{ (w_{n+1}, v_{n+1}) \}$.  
\item[$\exists$-extension] 
In this case, $\exists\in\sop$ and $@\in\sop$.
We show that there exists a forth
$\exists$-extension of $h$; 
for the back
$\exists$-extension, the arguments are symmetric.  
Let $w_{n+1}$ be any state from $\str{M}$.  
Let us consider a move along the edge 
$\Delta[\overline{z_n}] \xrightarrow{\exists} \Delta[\overline{z_{n+1}}]$.
We get $(\str{M}\ext{\overline{z_{n+1}}}{\overline{w_{n+1}}},w_n)
\approx_\omega (\str{N}\ext{\overline{w_{n+1}}}{\overline{v_{n+1}}},v_n)$.  
Then make another move along $\Delta[\overline{z_{n+1}}] \xrightarrow{\at{z_{n+1}}} \Delta[\overline{z_{n+1}}]$.  
It follows that $(\str{M}\ext{\overline{z_{n+1}}}{\overline{w_{n+1}}},w_{n+1})\approx_\omega (\str{N}\ext{\overline{z_{n+1}}}{\overline{v_{n+1}}},v_{n+1})$.  
Let $g : \str{M} \nrightarrow \str{N}$ be the basic partial isomorphism $h\cup\{(w_{n+1},v_{n+1})\}$.
\end{description}
One can start with two pointed models $(\str{M},w_1)$ and $(\str{N},v_1)$
such that $(\str{M}, w_1) \approx_\omega (\str{N},v_1)$, define a basic partial
isomorphism $h:\str{M} \nrightarrow \str{N}$ by $h(w_1) = v_1$ and then extend it as described above to build a back-and-forth system $\bfsys$ between $\str{M}$ and $\str{N}$.
\end{proof}

Note that \cref{th:BF-EF}(\ref{th:BF-EF-1}) applies to $\HDPL$, $\HPL$ and $\DPL$, while \cref{th:BF-EF}(\ref{th:BF-EF-2}) applies to $\HDPL$ and $\HPL$ but not to $\DPL$ (since it
requires closure under store). 

The following example shows that the assumption on closure under retrieve in
presence of possibility is necessary for \cref{th:BF-EF}(2).   
\begin{example}\label{ex:pos}
Let $\Delta$ be a signature with no nominals, one binary relation symbol
$\lambda$ and one propositional symbol $p$.  
Let $\str{M}$ and $\str{N}$ be the $\Delta$-models shown to the left and right, respectively, in the following
diagram.
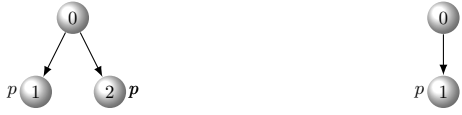
\begin{equation*}
\begin{tikzpicture}[baseline=0pt, sibling distance = 1.4cm, level distance=1.4cm, edge from parent/.style={draw,-latex},circle,scale=.7, transform shape]
\begin{scope}
\node [world] (O) {$0$}
  child{ node [world] (A1) {$1$}}
  child{ node [world] (A2) {$2$}};  
\node[left=-5pt of A1]{$p$};
\node[right=-5pt of A2]{$p$};
\end{scope}
\begin{scope}[xshift=7cm]
\node [world] (O) {$0$}
  child{ node [world] (A1) {$1$}};
\node[left=-5pt of A1]{$p$};
\node[right=-5pt of A2]{$p$};
\end{scope}
\end{tikzpicture}
\end{equation*}
If $\frag$ is obtained from $\HDPL$ by dropping $@$ and $\exists$, then it is straightforward to show that $(\str{M},0)\approx_\omega (\str{N},0)$.
Now let $h:\str{M}\nrightarrow \str{N}$ be the basic partial isomorphism defined by $h(0)=0$ and $h(1)=1$.
Then there is no forth $\Diamond$-extension of $h$ to $2$. 
Therefore, there is no back-and-forth system between $\str{M}$ and $\str{N}$.
\end{example}

The following example shows that the assumption on closure under retrieve in the
presence of existential quantifiers is necessary for \cref{th:BF-EF}(2). 

\begin{example}\label{ex:quant}
Let $\Delta$ be a signature with no nominals, one binary relation symbol
$\lambda$ and two propositional symbols $\{p,q\}$.  Let $\str{M}$ and $\str{N}$ be
the $\Delta$-models shown to the left and right, respectively, in the following
diagram.
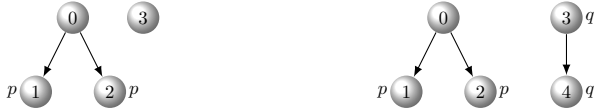
\begin{equation*}
\begin{tikzpicture}[baseline=0pt, sibling distance = 1.4cm, level distance=1.4cm, edge from parent/.style={draw,-latex},circle,scale=.7, transform shape]
\begin{scope}
\node [world] (O) {$0$}
  child{ node [world] (A1) {$1$}}
  child{ node [world] (A2) {$2$}};  
\node[left=-5pt of A1]{$p$};
\node[right=-5pt of A2]{$p$};
\node[world,right=20pt of O]{$3$};
\end{scope}
\begin{scope}[xshift=7cm]
\node [world] (O) {$0$}
  child{ node [world] (A1) {$1$}}
  child{ node [world] (A2) {$2$}};
\node[left=-5pt of A1]{$p$};
\node[right=-5pt of A2]{$p$};
\node [world,right=1.7cm of O] (R) {$3$}
  child{ node [world] (B) {$4$}};
\node[right=-5pt of R]{$q$};
\node[right=-5pt of B]{$q$};
\end{scope}
\end{tikzpicture}
\end{equation*}
If $\frag$ is obtained from $\HDPL$ by dropping $@$, then
$(\str{M},0)\exwins_\omega(\str{N},0)$ because the current state will always be
in the connected component of $0$ during the EF game (regardless of whether the
state $3$ is named or not, \Abelard cannot change the current state to $3$ in
the absence of retrieve). 
Now, assume a basic partial isomorphism $h:\str{M}\to\str{N}$ defined by
$h(i)=i$ for all $i\in\{0,\dots,2\}$.   
There is no forth $\exists$-extension of $h$ to $3$.
Therefore, there is no back-and-forth system between $\str{M}$ and $\str{N}$.
\end{example}

\subsection{Image-finite models and the Hennessy-Milner theorem} \label{sec:IFM}

The main result of this section is a Hennessy-Milner theorem.
Here, and indeed for the remainder of the paper,
we drop the first-order quantifiers and the constructors for actions. Thus,
we work within a quantifier-free fragment of $\HPL$ closed
under possibility, that is, $\Diamond\in\sop$. 

\begin{definition}[Image-finite model] \label{def:image-finite}
A model $\str{M}$ is \emph{image-finite} if each state has a finite number of direct successors, that is, $\lambda^{\str{M}}(w)$ is finite for all states $w$ in $\str{M}$ and all binary relation symbols $\lambda$ in the underlying signature.
A pointed model $(\str{M},w)$ is image-finite if $\str{M}$ is image-finite.
\end{definition}

We show that two image finite pointed models $(\str{M},w)$ and $(\str{N},v)$ are elementarily equivalent iff they are game equivalent w.\,r.\,t.\ all $\omega$-EF games.

\begin{theorem}\label{th:EE-OE}
Let $(\str{M},w)$ and $(\str{N},v)$ be two image-finite pointed models defined over a signature $\Delta$.
Then:
\begin{equation*}
(\str{M},w)\equiv(\str{N},v) \quad\text{iff}\quad (\str{M},w)\approx_\omega(\str{N},v)
\ \text{.}
\end{equation*}
\end{theorem}
\begin{proof}
For the backward implication, assume that $(\str{M},w)\approx_\omega(\str{N},v)$.
Since all sentences are of finite length,
\begin{equation} \label{eq:th:EE-OE}
(\str{M},w)\equiv(\str{N},v) \quad\text{iff}\quad (\str{M}\red_{\Delta_f},w)\equiv(\str{N}\red_{\Delta_f},v)\text{ for all finite signatures }\Delta_f\subseteq \Delta
\ \text{.}
\end{equation}
For all finite signatures $\Delta_f\subseteq \Delta$,
since $(\str{M},w)\approx_\omega(\str{N},v)$, we have 
$(\str{M}\red_{\Delta_f},w)\approx_\omega(\str{N}\red_{\Delta_f},v)$, which implies 
$(\str{M}\red_{\Delta_f},w) \exwins_{\tr}(\str{N}\red_{\Delta_f},v)$ for all gameboard trees $\tr$;
by \cref{cor:fh}, we obtain $(\str{M}\red_{\Delta_f},w)\equiv(\str{N}\red_{\Delta_f},v)$.
By \eqref{eq:th:EE-OE}, we get $(\str{M},w)\equiv(\str{N},v)$.

For the forward implication, first, we show that 
for all binary relations symbols $\lambda$ in
$\Delta$ and all states $w'\in\lambda^\str{M}(w)$,  
there exists a state $v'\in \lambda^\str{N}(v)$ such that $(\str{M},w') \eequiv (\str{N},v')$.
Assume that $\lambda^{\str{N}}(v)=\{v_1,\dots,v_n\}$ and let $w_1\in\lambda^\str{M}(w)$.
Suppose towards a contradiction that
for each $i\in\{1,\dots,n\}$ there exists a sentence $\phi_i$ such that
$(\str{M},w_1)\models\phi_i$ and $(\str{N},v_i)\not\models \phi_i$.  We have
$(\str{M},w)\models \pos{\lambda} \phi_1\wedge\dots\wedge \phi_n$ but
$(\str{N},v)\not\models \pos{\lambda} \phi_1\wedge\dots\wedge \phi_n$,
which is a contradiction with our assumptions.  Therefore, there exists
$i\in\{1,\dots,n\}$ such that for all sentences $\phi$, we have
$(\str{M},w_1)\models \phi$ iff $(\str{N},v_i)\models \phi$, that is,
$(\str{M},w_1) \eequiv (\str{N},v_i)$.

Secondly, we show that \Heloise has a winning strategy for the EF game played over a complete gameboard tree $\tr = \Delta(\xrightarrow{\lb_1} \tr_1, \xrightarrow{\lb_2} \tr_2,\dots)$ of height $\omega$.
Assume that \Abelard moves along an edge $\Delta \xrightarrow{\lb_i} \tr_i$.
We will show that \Heloise can match \Abelard's move such that the resulting pair of pointed models are elementarily equivalent.
Depending on the label $\lb_i$, we have three non-trivial cases:
\begin{description}
\item[{$\Delta \xrightarrow{\pos{\lambda}} \Delta$}] 
Assume that \Abelard's move along $\Delta \xrightarrow{\pos{\lambda}} \tr_i$ is $(\str{M}, w_i)$ where $w_i\in \lambda^{\str{M}}(w)$.  
By the first part of the proof, there exists $v_i \in \lambda^{\str{N}}(v)$ such that $(\str{M}, w_i) \eequiv (\str{N},v_i)$.
\item[$\Delta \xrightarrow{\at{k}} \Delta$] A move along this edge means that
\Abelard's pair becomes $(\str{M}, k^{\str{M}})$ and \Heloise's pair becomes
$(\str{N}, k^{\str{N}})$.  
Since $(\str{M}, w) \eequiv (\str{N}, v)$, we get $(\str{M}, k^{\str{M}}) \eequiv (\str{N}, k^{\str{N}})$.
\item[{$\Delta \xrightarrow{\downarrow} \Delta[x]$}] 
A move along this edge means that the new pair of pointed models consists of $(\str{M}\ext{x}{w}, w)$ and $(\str{N}\ext{x}{v}, v)$.  
This move gives name $x$ to the states $w$ and $v$.  
For all $\Delta[x]$-sentences $\phi$, we have 
$(\str{M}\ext{x}{w}, w)\models \phi$ 
iff
$(\str{M}, w)\models \store{x}\phi$ 
iff (since $(\str{M}, w)\eequiv(\str{N}, v)$)
$(\str{N}, v)\models \store{x}\phi$
iff
$(\str{N}\ext{x}{v}, v)\models \phi$.
Therefore, $(\str{M}\ext{x}{w}, w) \eequiv (\str{N}\ext{x}{v}, v)$.
\end{description}
In all three cases, $\tr_i$ is a gameboard tree and the pointed models obtained are elementarily equivalent.  
In particular, the resulting pointed models satisfy the same basic sentences.  
It follows that $(\str{M}, w) \exwins_{\tr} (\str{N}, v)$. 
Hence, $(\str{M}, w) \approx_\omega (\str{N}, v)$
\end{proof}

Image-finiteness is indeed necessary.
Recall the models $\str{M}$ and $\str{N}$ from \cref{ex:infinite-branch}.
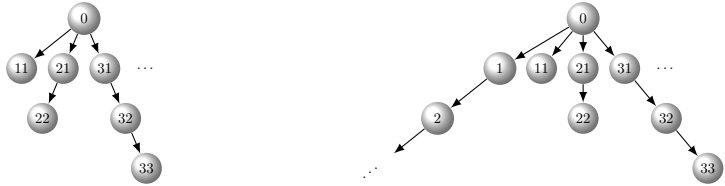
\begin{equation*}
\begin{tikzpicture}[baseline=0pt, sibling distance = 1cm, level distance=1.2cm, edge from parent/.style={draw,-latex},scale=.55, transform shape]
\begin{scope}
\node[circle,ball color=white,]{$~0~$}[grow=down]
  child{node[circle,ball color=white] {$11$}
        edge from parent
       }
  child{node[circle,ball color=white]{$21$}
        edge from parent
        child{ node[circle,ball color=white]{$22$} }
        child{ edge from parent [draw=none] }
       }
  child{node[circle,ball color=white]{$31$}
        edge from parent
        child{ edge from parent [draw=none] }
        child{ node[ball color=white,circle] {$32$}
               edge from parent
               child{ edge from parent [draw=none] }
               child{node[circle,ball color=white]{$33$}}
             } 
       }
  child{node{$\dots$}
        edge from parent[draw=none]
       };
\end{scope}
\begin{scope}[xshift=12cm]
\node[circle,ball color=white, fill=gray!10]{$~0~$}[grow=down]
  child{node[circle,ball color=white, fill=gray!10]{$~1~$}
        edge from parent
        child{node[circle,ball color=white, fill=gray!10]{$~2~$}
              edge from parent
              child{ node{\reflectbox{~~$\ddots$}}}
              child{ edge from parent [draw=none] }
              child{ edge from parent [draw=none] }
              child{ edge from parent [draw=none] }
             }
        child{edge from parent [draw=none]}
        child{edge from parent [draw=none]}
        child{edge from parent [draw=none]}     
       }
  child{node[circle,ball color=white, fill=gray!10] {$11$}
        edge from parent
       }
  child{node[circle,ball color=white, fill=gray!10]{$21$}
        edge from parent
        child{ node[circle,ball color=white, fill=gray!10]{$22$} }
       }
  child{node[circle,ball color=white, fill=gray!10]{$31$}
        edge from parent
        child{ edge from parent [draw=none] }
        child{ edge from parent [draw=none] }
        child{ node[ball color=white, fill=gray!10,circle] {$32$}
               edge from parent
               child{ edge from parent [draw=none] }
               child{ edge from parent [draw=none] }
               child{node[circle,ball color=white, fill=gray!10]{$33$}}
             } 
       }
  child{node{$\dots$}
        edge from parent[draw=none]
       };
\end{scope}
\end{tikzpicture}
\end{equation*}
$(\str{M},0)$ and $(\str{N},0)$ are elementarily equivalent, but they are not $\omega$-game equivalent.
\Abelard can win the EF game played over a complete gameboard tree of height $\omega$.
\Abelard's moves are depicted in the following diagram.
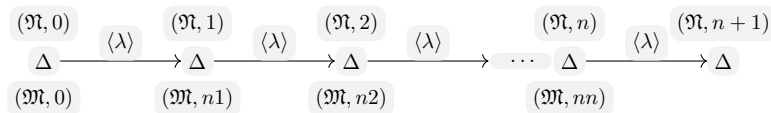
\begin{equation*}
\begin{tikzpicture}[baseline=0pt, scale=.8, transform shape]
\node[model, above=2pt of A0]{$(\str{N},0)$};
\node(A0)[sig]{$\Delta$};
\node[model, below=2pt of A0]{$(\str{M},0)$};

\node(A1)[sig, right=2cm of A0]{$\Delta$};
\node[model, above=2pt of A1]{$(\str{N},1)$};
\draw[->] (A0) to node[sig,above]{$\pos{\lambda}$} (A1) ;
\node[model, below=2pt of A1]{$(\str{M},n1)$};
\node(A2)[sig, right=2cm of A1]{$\Delta$};
\node[model, above=2pt of A2]{$(\str{N},2)$};
\draw[->] (A1) to node[sig,above]{$\pos{\lambda}$} (A2);
\node[model, below=2pt of A2]{$(\str{M},n2)$};
\node(A3)[sig, right=2cm of A2]{~~$\dots$~~};
\draw[->] (A2) to node[sig,above]{$\pos{\lambda}$} (A3);
\node(A4)[sig, right=0cm of A3]{$\Delta$};
\node[model, above=2pt of A4]{$(\str{N},n)$};
\node[model, below=2pt of A4]{$(\str{M},nn)$};
\node(A5)[sig, right=2cm of A4]{$\Delta$};
\node[model, above=2pt of A5]{$(\str{N},n+1)$};
\draw[->] (A4) to node[sig,above]{$\pos{\lambda}$} (A5);
\end{tikzpicture}
\end{equation*}
First, \Abelard moves $(\str{N},0)$ along $\Delta\xrightarrow{\pos{\lambda}}\Delta$ to obtain $(\str{N},1)$.
A move from \Heloise results in $(\str{M},n1)$, where $n$ is a natural number greater than $0$.
Then \Abelard moves $(n-1)$-times the pointed model $(\str{N},1)$  along  $\Delta\xrightarrow{\pos{\lambda}}\Delta$ to obtain $(\str{N},n)$.
\Heloise can move $(n-1)$-times the pointed model $(\str{M},n1)$  along  $\Delta\xrightarrow{\pos{\lambda}}\Delta$ to get $(\str{M},nn)$.
For the final round, \Abelard takes $(\str{N},n)$ along $\Delta\xrightarrow{\pos{\lambda}}\Delta$ to obtain $(\str{N},n+1)$ while \Heloise cannot match this move.

The following result is a corollary of \cref{th:bis-ef} and \cref{th:EE-OE}.

\begin{corollary}[Hennessy-Milner theorem] \label{cor:EE-BS}
Let $(\str{M},w)$ and $(\str{N},v)$ be two image-finite pointed models.
Then:
\begin{equation*}
(\str{M},w)\equiv(\str{N},v) \quad\text{iff}\quad (\str{M},w) \equiv_B (\str{N},v)
\ \text{.}
\end{equation*}
\end{corollary}
\begin{lemma} \label{lemma:pi}
Assume that $\{\Diamond,@,\downarrow\}\subseteq \sop$. 
Let $\str{M}$ and $\str{N}$ be two Kripke structures defined over a signature $\Delta$.
Let $B$ an $\omega$-bisimulation between $\str{M}$ and $\str{N}$ such that $(\overline{w},w)\mathrel{B_\ell}(\overline{v},v)$.
Then $h:\str{M}\nrightarrow\str{N}$ defined by $h(\overline{w}(i))=\overline{v}(i)$ for all $i\in\{1,\dots,\ell\}$ is a partial isomorphism.
\end{lemma}
\begin{proof}
Let $i,j\in\{1,\dots,\ell\}$ be two arbitrary natural numbers.
Since  $(\overline{w},w)\mathrel{B_\ell}(\overline{v},v)$, by (atv), $(\overline{w},\overline{w}(i))\mathrel{B_\ell}(\overline{v},\overline{v}(i))$.
\begin{itemize}
\item  By (wvar), $\overline{w}(i)=\overline{w}(j)$ iff  $\overline{v}(i)=\overline{v}(j)$.
Hence, $h$ is a partial bijection. 
\item By (prop) and (nom), $(\str{M},\overline{w}(i)) \models \rho$ iff  $(\str{N},\overline{v}(i)) \models \rho$ for all $\rho\in\Sen_b(\Delta)$.
\item If $\overline{w}(i) \mathrel{\lambda^\str{M}} \overline{w}(j)$ then
since $(\overline{w},\overline{w}(i))\mathrel{B_\ell}(\overline{v},\overline{v}(i))$, by (forth), $(\overline{w},\overline{w}(j))\mathrel{B_\ell}(\overline{v},v')$ for some state $v'\in\lambda^\str{N}(\overline{v}(i))$;
by (wvar), $v'=\overline{v}(j)$, which means $\overline{v}(i) \mathrel{\lambda^\str{N}} \overline{v}(j)$.

Similarly, if $\overline{v}(i) \mathrel{\lambda^\str{N}} \overline{v}(j)$ then by (back) and (wvar),  $\overline{w}(i) \mathrel{\lambda^\str{M}} \overline{w}(j)$.\qedhere
\end{itemize}
\end{proof}
\begin{definition}[Rooted pointed model] \label{def:rooted}
A pointed model $(\str{M},w)$  defined over a signature $\Delta = ((F, P), \Prop)$ is \emph{rooted} 
if $v \in (\bigcup_{\lambda \in P} \lambda^\str{M})^*(w)$ for all possible worlds $v\in|\str{M}|$.
\end{definition}
\begin{theorem}\label{th:rooted} 
Assume that $\{\Diamond,@,\downarrow\}\subseteq\sop$.
Let $(\str{M},w_0)$ and $(\str{N},v_0)$ be two rooted pointed models that are countable.
Then: 
\begin{equation*}
(\str{M},w_0) \cong (\str{N},v_0) \quad\text{iff}\quad 
(\str{M},w_0) \equiv_B (\str{N},v_0)
\ \text{.}
\end{equation*}
\end{theorem}
\begin{proof}
Obviously, any isomorphic pointed models are $\omega$-bisimilar.
For the backward implication, let $\{ w_i \mid i < \omega\}$ and $\{ v_i\mid i< \omega\}$ be enumerations of all possible worlds of $\str{M}$ and $\str{N}$, respectively, such that $w_i$ and $v_i$ are the direct successor of some state in $W_i=\{w_j\mid j< i\}$ and $V_i=\{v_j\mid j< i\}$, respectively, for all ordinals $i<\omega$.
We define a family of bisimilar correspondences $\{ (\overline{\mu_\ell},\mu_\ell)\mathrel{B_\ell} (\overline{\nu_\ell},\nu_\ell) \mid \ell\in\omega\}$ by induction:
\begin{description}
\item [($\ell=0$)] We have $ w_0 \mathrel{B_0}v_0$.
Let $\mu_0\coloneqq w_0$ and $\nu_0\coloneqq v_0$.
\item [($2\ell\Rightarrow 2\ell+1$)]
We assume  that  $(\overline{\mu_{2\ell}},\mu_{2\ell})\mathrel{B_{2\ell}} (\overline{\nu_{2\ell}},\nu_{2\ell})$ is defined such that
\begin{enumerate}[(a)] 
\item~$w_j$ is in the sequence $ \overline{\mu_{2\ell}}\,\mu_{2\ell}$ for all $j\in\{0,\dots,\ell\}$, and
\item~$v_j$ is in the sequence $\overline{\nu_{2\ell}}\,\nu_{2\ell}$ for all $j\in\{0,\dots,\ell\}$.
\end{enumerate}
By (st), we get $(\overline{\mu_{2\ell}}\,\mu_{2\ell}, \mu_{2\ell})\mathrel{B_{2\ell+1}} (\overline{\nu_{2\ell}}\,\nu_{2\ell},\nu_{2\ell})$. 
The goal is to add $w_{\ell+1}$ to the bisimilar correspondence.
Let $\overline{\mu_{2\ell+1}}\coloneqq \overline{\mu_{2\ell}}\,\mu_{2\ell}$ and $\overline{\nu_{2\ell+1}}\coloneqq \overline{\nu_{2\ell}}\,\nu_{2\ell}$.
Notice that $w_{\ell+1}$ is a successor of some state $w_j\in W_{\ell+1}$, which is included in the sequence $\overline{\mu_{2\ell+1}}$.
Therefore, $w_j=\overline{\mu_{2\ell+1}}(i)$ for some $i\in\{1,\dots,2\ell+1\}$.
By (atv), $(\overline{\mu_{2\ell+1}},\overline{\mu_{2\ell+1}}(i))\mathrel{B_{2\ell+1}} (\overline{\nu_{2\ell+1}},\overline{\nu_{2\ell+1}}(i))$.
By (forth), $(\overline{\mu_{2\ell+1}},w_{\ell+1})\mathrel{B_{2\ell+1}} (\overline{\nu_{2\ell+1}},\nu_{2\ell+1})$.
Let $\mu_{2\ell+1}\coloneqq w_{\ell+1}$.
\item [($2\ell+1\Rightarrow 2\ell+2$)]
By (st), we get $(\overline{\mu_{2\ell+1}}\,\mu_{2\ell+1}, \mu_{2\ell+1})\mathrel{B_{2\ell+2}} (\overline{\nu_{2\ell+1}}\,\nu_{2\ell+1},\nu_{2\ell+1})$.
The goal is to add $v_{\ell+1}$ to the bisimilar correspondence.
Let $\overline{\mu_{2\ell+2}}\coloneqq \overline{\mu_{2\ell+1}}\,\mu_{2\ell+1}$ and $\overline{\nu_{2\ell+2}}\coloneqq \overline{\nu_{2\ell+1}}\,\nu_{2\ell+1}$.
Notice that $v_{\ell+1}$ is a successor of some state $v_j$ in $V_{\ell+1}$, which is included in the sequence $\overline{\nu_{2\ell+2}}$.
Therefore, $v_j=\overline{\nu_{2\ell+2}}(i)$ for some $i\in\{1,\dots,2\ell+2\}$.
By (atv), we get $(\overline{\mu_{2\ell+2}},\overline{\mu_{2\ell+2}}(i))\mathrel{B_{2\ell+2}} (\overline{\nu_{2\ell+2}},\overline{\nu_{2\ell+2}}(i))$.
By (back), $(\overline{\mu_{2\ell+2}},\mu_{2\ell+2})\mathrel{B_{2\ell+2}} (\overline{\nu_{2\ell+2}},v_{\ell+1})$.
Let $\nu_{2\ell+2}\coloneqq v_{\ell+1}$.
\end{description}
By \cref{lemma:pi}, for each $\ell\in\omega$, $h_\ell:\str{M}\nrightarrow\str{N}$ defined by $h_\ell(\overline{\mu_\ell}(i))=\overline{\nu_\ell}(i)$ for all $i\in\{1,\dots,\ell\}$ is a partial isomorphism.
Since for all $\ell\in\omega$,
$\dom(h_\ell)\subseteq \dom(h_{\ell+1})$,
the union $h=\bigcup_{\ell\in\omega} h_\ell$ is well-defined and it is a partial isomorphism.
Since for all $\ell\in\omega$,  $w_\ell\in\dom(h_{2\ell}) $ and $v_\ell\in\dom(h_{2\ell+1})$,
$h$ is an isomorphism.
\end{proof}

The following result is a consequence of \cref{cor:EE-BS} and \cref{th:rooted}.

\begin{corollary}
Assume that $\{\Diamond,@,\downarrow\}\subseteq\sop$.
Let $(\str{M},w)$ and $(\str{N},v)$ be two rooted image-finite pointed models.
Then: 
\begin{equation*}
(\str{M},w_0) \cong (\str{N},v) \quad\text{iff}\quad 
(\str{M},w) \equiv (\str{N},v)
\ \text{.}
\end{equation*}
\end{corollary}

Notice that \cref{th:rooted} and \cref{cor:EE-BS}  hold also in $\HDPL$, since both elementary equivalence and $\omega$-bisimulation are stronger in $\HDPL$ than in $\HPL$.

\section{Conclusions and Future Work}\label{sec:conclusions}
We presented a novel notion of EF games for hybrid-dynamic propositional logic and its fragments.
Gameboard trees were introduced out of necessity when seeking a solution to accommodate moves in hybrid-dynamic EF games corresponding to different types of quantification. The best solution we could find (likely the cleanest and most elegant) was based on gameboard trees. 
Hybrid-dynamic EF games utilizing gameboard trees are applicable to other types of quantification, such as those related to temporal operators studied in \cite{KERNBERGER2020362}.  
Via gameboard sentences, the finite hybrid-dynamic EF games thus capture elementary equivalence and characterize their corresponding language fragment precisely. 
We showed that the existence of a winning strategy for the countably infinite EF game coincides with the existence of an $\omega$-bisimulation and, under certain conditions, a back-and-forth system between the structures. 
The parametric, syntactic approach of designing the moves in  hybrid-dynamic EF games should be applicable also to other variants of hybrid-dynamic logic, like its integration with branching-time logics~\cite{KERNBERGER2020362} or with past operators~\cite{abramsky_et_al:LIPIcs.MFCS.2022.7} or rigid symbols~\cite{GAINA2023103212,dgt-aiml2022}.
A future direction of research involves generalizing the current approach to EF games to an institutional setting, utilizing category theory, as exemplified in~\cite{gai-fraisse}. 
The institution-independent framework could be instantiated for other hybrid-dynamic logics such as Hybrid-Dynamic First-Order Logic with rigid symbols~\cite{GAINA2023103212} or Hybrid-Dynamic First-Order Logic with user-defined sharing. 
The hybrid counterpart of the latter logic~\cite{dia-msc} serves as the underlying logic of the H tool; refer to \cite{cod-H} for further details.

\section*{Acknowledgements}
D. G\u{a}in\u{a} has been partially supported by the Japan Society for the Promotion of Science under grant number 23K11048.
T. Kowalski gratefully acknowledges support from the Polish National Science Centre (NCN) grant OPUS-LAP no. K/NCN/000231. 

\bibliographystyle{plainurl}
\bibliography{hybrid-games}
\end{document}